\documentclass[12pt,a4paper]{article}
\usepackage[english]{babel}

\usepackage{amsmath,amsfonts,amssymb,amsthm}

\newcommand{\cK}{\mathcal{K}}
\newcommand{\cB}{\mathcal{B}}

\newcommand{\cH}{\mathcal{H}}

\newcommand{\cE}{\mathcal{E}}
\newcommand{\cD}{\mathcal{D}}

\newcommand{\cO}{\mathcal{O}}
\newcommand{\cF}{\mathcal{F}}
\newcommand{\sfh}{\mathsf{h}}
\newcommand{\sh}{\mathsf{h}}
\newcommand{\bS}{\mathbf{S}}
\newcommand{\bsh}{\cB(L^{2,s}(\bR^d),\sh)}
\newcommand{\bN}{\mathbf{N}}

\newcommand{\ip}[2]{\langle{#1},{#2}\rangle}
\newcommand{\pac}{P_{\rm ac}(H)}

\newtheorem{theorem}{Theorem}[section]
\newtheorem{proposition}[theorem]{Proposition}

\newtheorem{lemma}[theorem]{Lemma}

\theoremstyle{definition}
\newtheorem{assumption}[theorem]{Assumption}
\newtheorem{remark}[theorem]{Remark}
\newtheorem{definition}[theorem]{Definition}

\numberwithin{equation}{section}

\newcommand{\bR}{\mathbf R}
\newcommand{\bC}{\mathbf C}
\newcommand{\jap}[1]{\langle{#1}\rangle}
\newcommand{\ket}[1]{\vert{#1}\rangle}
\newcommand{\bra}[1]{\langle{#1}\vert}
\newcommand{\ve}{\varepsilon}
\DeclareMathOperator{\re}{Re}
\DeclareMathOperator{\im}{Im}

\newcommand{\abs}[1]{\lvert{#1}\rvert}
\newcommand{\norm}[1]{\lVert{#1}\rVert}

\newcommand{\rg}{\rho_{\gamma}}
\newcommand{\rmg}{\rho_{-\gamma}}

\begin{document}

\title{Metastable states when the Fermi Golden Rule constant
vanishes\footnote{\copyright~2013 by the authors. This paper may be reproduced, in its entirety, for non-commercial
purposes.}}
\author{
Horia D. Cornean\\
Department of Mathematical Sciences, Aalborg University\\
 Fr. Bajers Vej
7G, DK-9220 Aalborg \O, Denmark\\
E-mail: \texttt{cornean@math.aau.dk}\\[5pt]
Arne Jensen\\
Department of Mathematical Sciences, Aalborg University\\
 Fr. Bajers Vej
7G, DK-9220 Aalborg \O, Denmark\\
E-mail: \texttt{matarne@math.aau.dk}\\[5pt]
Gheorghe Nenciu\\
 Institute of Mathematics of the Romanian Academy, Research Unit 1\\ Calea Grivitei 21, RO-010702 Bucharest, Romania\\
 E-mail: \texttt{Gheorghe.Nenciu@imar.ro}}
\date{}
\maketitle

\begin{abstract}
Resonances appearing by perturbation of embedded non-degenerate
eigenvalues are studied in the case when the Fermi Golden Rule constant 
vanishes. Under appropriate smoothness properties for the resolvent of the unperturbed
Hamiltonian, it is proved that the first order Rayleigh-Schr\"odinger expansion exists. 
The corresponding metastable states are constructed using this
truncated expansion. We show that their exponential decay law has 
both the decay rate and the error term of order $\ve^4$,  where $\ve$ is the perturbation strength.
\end{abstract}

\section{Introduction}\label{section1}
In this paper we continue our study of the decay laws for resonances produced by perturbation of eigenvalues embedded 
in the continuous spectrum. 
More precisely, one considers an unperturbed Hamiltonian $H$ having a
non-degenerate eigenvalue $E_0$ embedded in its continuous spectrum. The degenerate case is by far 
more complicated and will be not discussed in this paper; we send the
reader 
to 
\cite{JN4,Or,Hu,MS}
and 
references therein for the  results known in this case.

Our  problem is to study the fate of the unperturbed state $\Psi_0$
corresponding to $E_0$ when adding a perturbation $W$  of strength
$\ve \ll 1$ so 
that the Hamiltonian becomes $H_{\ve}=H+\ve W$. The answers
are quite different depending on whether the unperturbed eigenvalue is
situated near an energetic threshold or far 
away inside 
the continuous spectrum. 

In the present paper we 
revisit the case of properly embedded eigenvalues, while the previous papers in this series 
(see \cite{DJN1} and references therein)
were mainly concerned with the threshold case for which there were no 
rigorous results available (however, see \cite{JN4,JN5} and Section 4 in 
\cite{JN2} for results concerning properly embedded eigenvalues). 

The problem of the decay laws for resonances in general and, in
particular, for resonances produced by perturbation of 
eigenvalues embedded in the continuous spectrum, has a distinguished
and ramified history ranging from experimental to rigorous 
levels, see e.g. \cite{CS,FGR,Gr,LZMM,NNP,Da,D,Ex,
N,Nic,CGH,Hu,Or,JN2,MS,SW,KR,H}, and
references given there. As is well known, the notion of `resonance' occurs often. 
It has many definitions and its meaning depends upon the context. 
For
example, in spectral and scattering theory a resonance is a complex number which may be a 
pole in the analytic continuation of the 
resolvent of the corresponding Hamiltonian, or an eigenvalue of the
dilated Hamiltonian. There is a huge literature about the subject,
both at  the mathematical level and at the physical level. We refer to \cite{H} for a comprehensive review.

The scope of our paper is limited. We restrict ourselves to the 
perturbative setting described above and we are only interested in dynamical
aspects in a Hilbert space $\cH$. In what follows, by a resonance (probably `metastable
state' is a better name) we shall understand a {\em pair} $(\Psi_{\ve},E_{\ve})$ such that 
$\Psi_{\ve} \in \cH$, $\|\Psi_{\ve}\|=1$ and $\lim_{\ve\to 0}\|\Psi_{\ve}-\Psi_{0}\|=0$ 
(`resonance eigenfunction'),  
and $E_{\ve} \in \bC$ with $\im E_{\ve} \leq 0$, (`resonance position') 
satisfying with some
accuracy the exponential decay law for the survival amplitude:
\begin{equation}\label{A}
\langle \Psi_{\ve}, e^{-itH_{\ve}} \Psi_{\ve} \rangle \simeq e^{-itE_{\ve}}.
\end{equation}
If the bound state survives after turning on the perturbation, then the resonance pair is given by the  
corresponding bound state eigenfunction and eigenvalue 
of 
$H_{\ve}$ for which  
 equality is realized in \eqref{A}. If the eigenvalue disappears for $\ve> 0$, then the situation is by far less clear. First of all,
as is well known, the semi-boundedness of $H_{\ve}$ forbids the equality in \eqref{A} so we are left with the problem of finding 
$(\Psi_{\ve},E_{\ve})$ such that:
\begin{align}\label{A<}
\sup_{t\geq 0}|\langle \Psi_{\ve}, e^{-itH_{\ve}} \Psi_{\ve} \rangle - e^{-itE_{\ve}}| &\leq \delta (\ve),\\
\lim_{\ve \rightarrow 0}\delta(\ve)=0,\quad 
\lim_{\ve\to 0}\|\Psi_{\ve}-\Psi_{0}\|&=0.\label{Abis}
\end{align}
Clearly, \eqref{A<} and \eqref{Abis} do not define the pair $(\Psi_{\ve},E_{\ve})$ uniquely and  this adds to the difficulty of the subject;
for the moment, the best one can do is to  search for pairs $(\Psi_{\ve},E_{\ve})$ 
leading to a $\delta(\ve)$ as small as possible.

A natural candidate for $\Psi_{\ve}$ is just the unperturbed eigenvector $\Psi_0$. 
With the exception of \cite{Hu} all the existing rigorous results are related to the (quasi-)exponential decay law for
$\langle \Psi_{0}, e^{-itH_{\ve}} \Psi_{0} \rangle $, at least as far as we know.
The story  started in the early days of quantum mechanics with 
the  computation by Dirac  of the decay rate in second order
  time-dependent perturbation theory, leading 
to the well known exponential decay law,
$e^{-2\ve^2\Gamma t}$, for the survival probability. Here $\Gamma$ is
given by the famous Fermi Golden Rule (FGR) constant: 
\begin{equation}\label{Gammy}
\Gamma \sim
|\langle \Psi_0,  W\Psi_{{\rm cont}, E_0}\rangle|^2, 
\end{equation}
 where  $\Psi_{{\rm cont},E_0}$ is a generalized eigenfunction
 corresponding to $E_0$ in the continuous spectrum (assumed to have multiplicity one). The FGR formula
 has been so influential that the common wisdom
 in theoretical physics is that the decay law for the
resonances produced by perturbation of non-degenerate bound states is
exponential. However, since the decay law cannot be exactly exponential 
at the rigorous level (for semi-bounded Hamiltonians), the crucial problem
is the estimation of the errors. This proved to be a hard problem, and
only during the past decades consistent rigorous results have been
obtained. It turns out that (see \cite{CGH,CS,
Hu,JN2, KR, MS, Or, SW} and the references given there) the
decay law is indeed \mbox{(quasi-)exponential}, i.e. exponential up to error terms
vanishing in the limit $\ve \rightarrow 0$, if  the resolvent of
the unperturbed Hamiltonian is sufficiently smooth, when projected onto
the subspace orthogonal to the eigenvalue under consideration. First of all, in the dilation analytic
setting of the Balslev-Combes theory \cite{BC} there is a mathematically
well defined  candidate for the resonance position, $E_{\ve}$,
namely the perturbed eigenvalue of the dilated Hamiltonian. In this context Hunziker \cite{Hu} proved that
\begin{equation}\label{A_0}
|\langle \Psi_0, e^{-itH_{\ve}} \Psi_0 \rangle -e^{-itE_{\ve}}| \lesssim \ve^2.
\end{equation}
Since dilation analyticity is a strong assumption, much effort has
been devoted to the extension of the above result to 
the case when analyticity is relaxed to some smoothness conditions
(see \cite{CGH,CS,
JN2, MS, Or, SW}  and the references given there). More precisely,
if the resolvent of
the unperturbed Hamiltonian is sufficiently smooth, when projected onto
the subspace orthogonal to the spectral subspace corresponding to the eigenvalue under consideration, it  has been proved
that
one can {\em find} $E_{\ve}$, $\im E_{\ve}\leq 0$, such that 
\eqref{A_0} holds true, see Theorem 4.2ii in \cite{JN2} and its slight refinement in the present Section~\ref{section2}.
 Moreover, it turns out that if $\Gamma$ 
given by the FGR is nonzero, then  $\im E_{\ve}=-\ve^2\Gamma +
 \cO(\ve^3)$, which is consistent with the FGR formula.  The discussion
 in the next paragraphs strongly suggests that 
 the error term in \eqref{A_0} is optimal with respect to the power of $\ve$.

The  problem considered here is  whether the error term can
be made \emph{smaller} by choosing a better ansatz for the initial state by replacing 
$\Psi_0$ with a properly chosen, $\ve$-dependent, resonance
eigenfunction. 
In the Balslev-Combes dilation analytic setting this question has been
 addressed already by Hunziker \cite{Hu}. In that context
(see \cite{Si})   the resonance position has a clear-cut
spectral meaning. Hunziker proved that {\em if} $E_0$ {\em is 
isolated} and the formal Rayleigh-Schr\"{o}dinger (R-S) perturbation expansion for the perturbed
eigenfunction is well defined up to order $\ve^N$ (as it is the case for the atomic Stark effect), 
then by using as the resonance eigenfunction the normalized truncated R-S series, one can improve the error term to be of 
order $\ve^{2N+2}$. Here we also have $|\im E_{\ve}| \lesssim
\ve^{2N+2}$. For the embedded case (even in the dilation analytic 
case), the problem of improving the error term by choosing a better
ansatz for the initial state remained open. Due to Hunziker's results,
the natural conjecture is that under appropriate smoothness
conditions, the existence of the formal R-S 
perturbation expansion up to order $N$ should lead to an exponential
decay law with a smaller error term. For the case of embedded
eigenvalues the R-S series generally breaks down already at order $N=1$. Hence there are two
 problems to be solved. The first one is to 
  seek conditions under which the R-S perturbation expansion exists up
  to some order $N\geq 1$ and then to construct the `corrected'
  resonance eigenfunction. 
The second (harder) one is to prove that
under appropriate smoothness conditions, the new error term is indeed smaller.

The main result of the paper is a positive answer to both questions 
 for $N=1$ in the R-S expansion, see Proposition~\ref{bdd} and Theorem~\ref{ve4} below. 
More precisely,  suppose $\Gamma$ as given by the FGR vanishes, while
the second derivative with respect to the
energy of the generalized eigenfunction(s) of 
the unperturbed Hamiltonian exists in a neighborhood of $E_0$ and is $\theta$-H\"older continuous with some
$\theta >0$. Then the formal R-S perturbation expansion 
exists to order $N=1$, and for the corresponding initial value 
one can prove an exponential decay similar to \eqref{A_0} with both decay rate, $\im E_{\ve}$, and  error term of 
order $\ve^4$ or smaller. 

The contents of the paper is as follows. In Section \ref{section2} we give the main
results with an outline of proofs. Section \ref{section3} contains 
the technical details. In Section \ref{section4} we present a class of two channel
Schr\"odinger operators for which our abstract theory applies. In two Appendices
we collect, in a form appropriate for us, some known facts about
H\"older properties of the Cauchy integral transform, and about resolvent smoothness and $\Gamma$-operator for one body
Schr\"odinger operators, respectively.

\section{The results and outline of proofs}\label{section2}
Throughout the paper `$s$ sufficiently small' is a shorthand for
`there exists $s_0 >0$ such that for $0<s<s_0$'. Also,  for $A,B \geq
0$, we write $A\lesssim B$ instead of `there exists a constant  $0<C
<\infty$, independent of $A$ and $B$, such that $A\leq C B$'.

Our results are model independent, in the sense that only the boundedness of the perturbation and some 
smoothness of the resolvent of the unperturbed Hamiltonian are
demanded. For example, in Theorem \ref{ve4} we  require that the second 
derivative with respect to the spectral parameter of the generalized
eigenfunctions of $H$ exists and is H\"older continuous. Therefore we  
develop the theory at the abstract level and verify these assumptions for each concrete application. 

Let $H$ be a self-adjoint 
operator in a Hilbert space $\cH$ and
$E_{0}$ a non-degenerate eigenvalue of $H$, while $P_0$ is the
corresponding orthogonal projection:
$$HP_{0}=E_{0}P_{0},\;\;\dim P_{0} =1,\quad P_0\Psi_0=\Psi_0,\quad
\|\Psi_0\| =1, \quad Q_0=1-P_0.$$ 

Without loss of 
generality we can take $E_0=0$ in what follows. 
We denote by $P(\Delta)$ the spectral measure of $H$. The first basic assumption is that except  eigenvalue zero 
$H$ only has absolutely continuous spectrum in some neighborhood of the origin: 

\begin{assumption}\label{Aac}
There exists $a>0$ such that  $J_a \cap
\sigma_{\rm pp}(H)=\{0\}$ and $J_a \cap
\sigma_{\rm sc}(H)=\emptyset$, where $J_a=(-a,a)$.
\end{assumption}

Then we add a perturbation $W$ of strength $\ve >0$, $\ve \rightarrow 0$, and consider the perturbed operator
\begin{equation}
H_{\ve}=H+\ve W.
\end{equation}
 In order to keep the technicalities at a reasonable level we impose:
\begin{assumption}\label{Bdd}
$W$ is self-adjoint and bounded.
\end{assumption}
Adding some supplementary conditions, one can extend the results of this
paper to the case when $W$ is only relatively compact 
with respect to $H$. The case of singular perturbations is much harder,
and detailed results are only known in specific cases, 
as for example the Stark effect; we shall not consider the case of singular perturbations here.

The most natural candidate for the resonance eigenfunction is $\Psi_0$. Thus the most studied object has been the
amplitude of the survival probability of the unperturbed eigenfunction:
\begin{equation}\label{A^0}
A^0(\ve,t):=\langle \Psi_0,e^{-itH_{\ve}}\Psi_0\rangle.
\end{equation}
Let us first discuss the trivial case when $E_0=0$ is isolated and
lies in
the resolvent set of $Q_0HQ_0$. In this case the natural 
choice for $E_{\ve}$ in \eqref{A_0} is just the perturbed eigenvalue:
\begin{equation}
H_{\ve}\Psi_{\ve}=E_{\ve}\Psi_{\ve},\;\; \| \Psi_{\ve} \| =1,\;\;\|\Psi_{\ve}-\Psi_0\|\lesssim \ve.
\end{equation}
Using  $(e^{\pm itH_{\ve}}-e^{\pm itE_{\ve}})\Psi_{\ve}=0$ we have:
\begin{equation}\label{Psi0Psi}
A^0(\ve,t)-e^{-itE_{\ve}}=\langle (e^{-itH_{\ve}}-e^{-itE_{\ve}})(\Psi_0 -\Psi_{\ve}),\Psi_0 -\Psi_{\ve}\rangle,
\end{equation}
which implies:
\begin{equation}\label{A0reg}
|A^0(\ve,t)-e^{-itE_{\ve}}| \leq 2 \| \Psi_0 -\Psi_{\ve}\|^2 \lesssim \ve^2.
\end{equation}

Due to the following result the
estimate in \eqref{A0reg} is optimal, if $\Psi_0$ is not replaced with a
better choice:   
\begin{proposition}\label{lowb}
 Suppose that for sufficiently small $\ve$ there exists exactly one (possibly
 embedded) eigenvalue $E_{\ve}$, while the singular continuous
 spectrum is empty. Assume that
 there exists $\Psi_1\neq 0$ such that $\|\Psi_{\ve}-\Psi_0-\ve
 \Psi_1\|=o(\ve)$ and $\langle \Psi_0,\Psi_1\rangle =0$. Then 
there exists $C >0$
 such that
 $$\sup_{t\geq 0}|A^0(\ve,t)-e^{-itE_{\ve}}|\geq C\ve^2.$$
\end{proposition}

 Note that the existence of  a $\Psi_1$ in the above Proposition
 is guaranteed, 
if $E_0=0$ is isolated and $\Psi_{\ve}$ is obtained by applying to
$\Psi_0$ the
Sz.-Nagy unitary between the unperturbed projection $P_0$ and the perturbed
one $P_{\ve}$ (see~\cite{kato}). In the embedded case proving that the eigenvalue
can survive is highly nontrivial~\cite{AHM}. 
In \cite{FMS} a class of perturbations $W$ is considered such that \emph{if the eigenvalue survives} then the FGR constant must be zero and $\Psi_1$ can be constructed.  

Under the assumptions of Proposition \ref{lowb}, if there is a vector $\Psi^N(\ve)$ with $\|\Psi^N(\ve)\|=1$ and $\|\Psi^N(\ve)-\Psi_{\ve}\|=\mathcal{O}(\ve^N)$, then the quantity 
\begin{equation}\label{A^N}
A^N(\ve,t):=\langle \Psi^N(\ve),e^{-itH_{\ve}}\Psi^N(\ve)\rangle 
\end{equation}
will obey 
$$\sup_{t\geq 0}|A^N(\ve,t)-e^{-itE_{\ve}}|\leq C\ve^{2N}.$$
Of course, when taking an eigenvector $\Psi_{\ve}$ as the initial value, the error term vanishes.

In the nontrivial cases, i.e. when either $E_0$ is embedded in the
continuous spectrum of $H$, or $W$ is singular with respect to $H$ as
in the Stark effect, the generic phenomenon is that 
$E_0$ is moved out of the real axis and becomes 
a resonance, $E_{\ve}$, with $\im E_{\ve}<0$. Again, one expects that 
$\Psi_0$ is a good candidate for the resonance eigenfunction i.e.
\begin{equation}\label{A0r}
|\langle \Psi_0, e^{-itH_{\ve}} \Psi_0 \rangle -e^{-itE_{\ve}}| \leq \delta (\ve),
\end{equation}
uniformly in time, with $\lim_{\ve \rightarrow 0}\delta(\ve)=0$. Proving \eqref{A0r} is much harder; 
in particular, in the general case there is no obvious candidate for
$E_{\ve}$. The situation is fully understood in  the dilation 
analytic case, where an analogue of the 
Kato-Rellich perturbation theory has been developed, see~\cite{Si}. In
particular, the resonance position, $E_{\ve}$, is 
unambiguously defined as an eigenvalue of the 
dilated Hamiltonian. In the dilation analytic setting Hunziker~\cite{Hu} 
proved that \eqref{A0r} holds true with $\delta(\ve) \sim \ve^2$, i.e. the error term has the same size as 
in the isolated eigenvalue case.

For the smooth case, i.e. when the resolvent of $Q_0HQ_0$ has smooth limit values  on the real axis in a neighborhood of $0$, the situation is again satisfactory, as \eqref{A0r}
with $\delta(\ve) \sim \ve^2$ was proved under fairly weak smoothness
conditions (see Assumption \ref{smG} below).

As an example, we give Theorem~\ref{JN06} below. It 
 is a slight 
improvement of Theorem 4.1ii in \cite{JN2}. For related results giving the same size of the error term, see~\cite{CGH,CS}. 

We use a factored form of the perturbation $W$, defined as follows.
\begin{assumption}
Assume there exist a Hilbert space $\cK$ and two bounded operators $A\in\cB(\cH,\cK)$ and $D\in\cB(\cK)$, such that $D$ is a self-adjoint involution and such that
\begin{equation}\label{fact}
W=A^{*}DA
\end{equation}
\end{assumption}

We note that this type of assumption is very flexible, since it allows us in the Schr\"{o}dinger operator case to consider a $W$ which is a sum of a multiplicative potential and a finite rank operator. Using finite rank operators it is easy to construct examples with eigenvalues embedded in the continuous spectrum. See \cite{JN-sinha} for some examples in the threshold case.

 Define
\begin{equation}\label{w}
G(z)=AQ_{0}(H-z)^{-1}Q_{0}A^{*},
\end{equation}
and
\begin{align}\label{F0z}
F^0(z,\ve)&=\ve \langle\Psi_{0},W\Psi_{0}\rangle-z \notag\\
&\quad-\ve^{2}\langle\Psi_{0},A^{*}D\{
G(z)
-\ve G(z)[D+\ve G(z)]^{-1}G(z)\}DA\Psi_{0}\rangle.
\end{align}

Notice that $F^0(z,\ve)=\overline{F^0(\overline z,\ve)}$. 
Then using the Stone formula, the Schur-Livsic-Feshbach-Grushin (SLFG)
formula, and the Kato-Rellich regular 
perturbation theory,  
one obtains the starting  formula for the stationary approach to the
decay law problem (see~\cite{JN2} for details and references):
\begin{equation}\label{Main0}
A^0(\ve,t)=
\lim_{\eta \searrow 0}\frac{1}{2\pi i}\int_{\bR}dx\, e^{-ixt}
\Bigl(\frac{1}{F^0(x+i\eta,\ve)}-\frac{1}{F^0(x-i\eta,\ve)}\Bigr).
\end{equation}

In justifying the r.h.s. of \eqref{F0z}, and  evaluating the r.h.s. of \eqref{Main0}, it is important to
ensure that $G(z)$ is uniformly bounded and smooth 
in the norm topology in the rectangle
\begin{equation}\label{defDa}
D_a=\{z=x+i\eta\in \bC\,|\, x\in J_a= (-a,a),\;0<\eta<1\}.
\end{equation}
Let $\omega\colon[0,\infty)\mapsto [0,\infty)$ be a modulus of continuity, i.e. continuous and increasing,
with $\omega(0)=0$. Let
$\omega_{\theta}(x)=x^{\theta}$,  $\theta \in (0,1)$,  
denote the H\"older modulus of continuity. We denote by $C^{n,\omega}(D_a;B)$ the vector-valued H\"{o}lder-type space. Its norm is given as follows.
Let $f\in C^{n,\omega}(D_a;B)$.
\begin{equation}
\norm{f}_{C^{n,\omega}(D_a;B)}=\sum_{k=0}^n\sup_{z\in D_a}\norm{f^{(k)}(z)}_B
+\sup_{\substack{z_1,z_2\in D_a\\ z_1\neq z_2}}\frac{\norm{f(z_1)-f(z_2)}_B}{\omega(z_1-z_2)}.
\end{equation}
Our condition on the family $F(\cdot,\ve)$ is that
\begin{equation}
\sup_{0<\ve\leq\ve_0}\norm{F(\cdot,\ve)}_{C^{n,\omega}(D_a;B)}<\infty.
\end{equation}
This in particular implies that
\begin{equation}\label{Cnomega}
\sup_{0<\ve\leq \ve_0}\| F^{(n)}(x+i\eta,\ve)-F^{(n)}(y+i\eta,\ve)\|_B \lesssim\omega(|x-y|).
\end{equation}
Our condition has as a consequence that boundary values $F(\cdot+i0,\ve)$ exist and belong to the space $C^{n,\omega}(J_a;B)$, whose definition is an obvious modification of the one for $C^{n,\omega}(D_a;B)$.

\begin{assumption}\label{smG} We have for $G(z)$ given by \eqref{w} that
 \begin{equation}
G(\cdot) \in C^{1,\omega}(D_a; \cB(\cK))\quad\text{and}\quad \int_0^1\frac{\omega(x)}{x}dx <\infty.
\end{equation}
\end{assumption}

\begin{theorem}\label{JN06}
Under  Assumption \ref{smG} and  $0<\ve<\ve_0$ taken  sufficiently
small we have $F^0(\cdot,\ve) \in C^{1,\omega}(D_a;\bC)$. 

In particular, this function has an extension to the real axis with the same
smoothness properties $ F^0(x,\ve) := \lim_{\eta \searrow
  0}F^0(x +
i\eta,\ve) \in C^{1,\omega}(J_a;\bC)$. 

Let $R^0(x,\ve)$ and $I^0(x, \ve)$ be the real and imaginary part of $F^0(x,\ve)$, respectively,
\begin{equation*}
F^0(x,\ve) =: R^0(x,\ve)+ i I^0(x, \ve).
\end{equation*}
For a fixed $\ve$ sufficiently small the equation
\begin{equation}\label{x0}
R^0(x,\ve)=0
\end{equation}
has a unique solution $x^0(\ve)$ in $J_{a/2}$, which obeys the estimate $|x^0(\ve)| \lesssim \ve$.
Define
\begin{equation}\label{E0ve}
E^0_{\ve}:=x^0(\ve)+iI^0(x^0(\ve),\ve).
\end{equation}
Then for sufficiently small $\ve$ we have:
\begin{equation}\label{A0main}
|A^0(\ve, t)-e^{-itE^0_{\ve}}| \lesssim \ve^2.
\end{equation}
\end{theorem}

\begin{remark} We would like to mention the following  facts:
\begin{itemize}
\item[(i)] The estimate \eqref{A0main} has exactly the same
form as \eqref{A0reg} 
for the   case of eigenvalues, with  the perturbed eigenvalue
replaced by the `resonance position', $E^0_{\ve}$. In particular, 
 Proposition~\ref{lowb}
shows that in general, the error in \eqref{A0main} cannot be made smaller.

\item[(ii)] The computation of $I^0(x^0(\ve),\ve)$, using $|x^0(\ve)|\lesssim \ve $, leads to
$I^0(x^0(\ve),\ve)= -\ve^2\Gamma_{FGR}+\cO(\ve^3)$, 
with 
\begin{equation}\label{GammaFRG}
\Gamma_{FGR}:=\pi \langle \Psi_0,W\delta(Q_0HQ_0)W\Psi_0\rangle,
\end{equation}
which coincides with the result given by the Dirac computation.  Notice that while in \cite{CGH,CS,
 MS, SW} an FGR condition is required (i.e. either $\Gamma_{FGR}>0$ or at least $I^0(x^0(\ve),\ve) <0$) 
no such 
condition appears in Theorem \ref{JN06}. We also note that Orth \cite{Or} proves that 
$I^0(x^0(\ve),\ve) =0$ if and only if $x^0(\ve)$ is an eigenvalue. The improvement in 
Theorem \ref{JN06} compared 
to Theorem 4.2ii in \cite{JN2} is that the smoothness condition $G(\cdot) \in C^{1,\omega_{\theta}}$ 
for some 
$\theta >0$ is weakened to Assumption \ref{smG} which appears to be optimal.

\item[(iii)] 
In the analytic case, the resonance position is spectrally defined as a pole of the analytically continued 
resolvent and coincides with the zero $z_r=x_r+iy_r$ of the analytic continuation of
$F^0(z,\ve)$. In the smooth case its definition also involves $\Psi_0$
since it is given via the limit values of $F^0(z,\ve)$.
 Comparing the decay law given by Theorem~\ref{JN06} with
the one given by Hunziker in the analytic case, one can show 
that $|E^0_{\ve} -z_r| \lesssim \ve^2|y_r|$,
i.e. up to some order in $\ve$,  $E^0_{\ve}$ is indeed a spectral object of the family
$H_{\ve}$, see~\cite{CGH} and \cite{JN5}.

\item[(iv)] 
As in \cite{JN2}  we consider here only the problem of uniform in time bounds for the error term in the 
exponential decay law for survival probability amplitude. One can consider a more general problem by taking an initial condition
of the form $g_{\ve}(H_{\ve})\Psi_0$ (where $g_{\ve}(\lambda)$ in an appropriate cut off function) and prove that
$$
e^{-itH_{\ve}}g_{\ve}(H_{\ve})\Psi_0= a(\ve)e^{-itE_{\ve}}\Psi_0 + R(\ve, t),
$$ 
where the dispersive part $\langle \Psi_0,R(\ve, t)\rangle$ decays in time polynomially or 
even almost exponentially fast, see \cite{ CGH, Hu,MS,SW,CS,KR}
for details and further references.
\end{itemize}
\end{remark}

We now turn to the main question considered in this paper: Can the error term  be made \emph{smaller} by choosing a better ansatz for 
the resonance eigenfunction? The heuristics put forward above indicates that in the  case of an isolated eigenvalue a better ansatz
should be related to the R-S expansion of the perturbed
eigenfunction. In the dilation analytic case Hunziker proved that if $E_0$ is {\em  isolated} (but becomes a resonance due to the singularity of the perturbation)
and the formal R-S expansion for the perturbed eigenfunction is well defined up to order $\ve^N$, 
then by using as a resonance function the normalized truncated R-S series,
one can prove that both $\im E_{\ve} $ 
and the error term are of order $\ve^{2N+2}$. In particular, for the atomic Stark effect $N$ can be taken arbitrarily 
large. It follows that in this case we have $|\im E_{\ve}| \lesssim \ve^p $
for any integer $p$, which is consistent with the known fact that the 
imaginary part of the resonance position is exponentially
small. Furthermore, the error term can be made smaller than any power of $\ve$.

The problem with the embedded eigenvalues (even in the analytic case)
is that generically, the R-S expansion for the perturbed 
eigenfunction is already ill defined for $N=1$. In order to be able to
make a better ansatz, one needs to find conditions for the existence of 
the R-S expansion up to some order $N \geq 1$. Our first result states that  under appropriate smoothness conditions
the vanishing of the FRG constant 
$\Gamma_{FGR}=\pi \langle \Psi_0,W\delta(Q_0HQ_0)W\Psi_0\rangle =0$ 
insures that the  R-S expansion for the perturbed 
eigenfunction  is well defined for $N=1$, see
Proposition \ref{bdd}. 

The main result of this paper  is a generalization to embedded 
eigenvalues and smooth setting of Hunziker's results for $N=1$, see Theorem \ref{ve4}.

Now we state the smoothness properties we need in order to introduce the Kuroda $\Gamma$ 
operator \cite{kuroda}. Keep in mind that only the smoothness in a neighborhood of the origin matters.

Let $P(\Delta)$ be the spectral measure of $H$, and
$Q(J_a):=P(J_a)-P_0$; in particular, $Q_0=Q(J_a)+P(\bR\setminus
J_a)$. We can write the decomposition:
$$
\cH=P_0\cH \oplus Q(J_a)\cH \oplus P(\bR\setminus J_a)\cH =:P_0\cH \oplus \cH_< \oplus \cH_>.
$$
In order to keep the notation and technicalities at a reasonable level, we supplement Assumption \ref{Aac} by:
\begin{assumption}\label{Aconst}
The multiplicity of the absolutely continuous spectrum of $H$ in $J_a$ is constant.
\end{assumption}
This assumption means that there exists a Hilbert space $\sfh$ and a unitary map
\begin{equation}
\widetilde{\Gamma}\colon Q(J_a)\cH\to L^2(J_a,\sfh),
\end{equation}
such that 
\begin{equation}
(\widetilde{\Gamma} H\widetilde{\Gamma}^{\ast} \psi)(\lambda)=\lambda \psi(\lambda),
\quad \psi(\lambda)\in \sfh, \quad \lambda\in J_a.
\end{equation}
We will always see $\widetilde{\Gamma}$ as a partial isometry extended by zero outside the range of $Q(J_a)$. In all the cases we consider,  $\widetilde{\Gamma}$ is constructed in the following way. We assume that there exists a dense subset $\cD\subset \cH$ and a family of operators 
$$ \Gamma(\lambda):\cD\mapsto \sfh,\quad \lambda\in J_a,$$
such that if $f\in \cD$ then the map 
$$J_a\ni \lambda \mapsto \Gamma(\lambda)f\in \sfh$$ 
is continuous.  For the free Laplacian, $\Gamma(\lambda)$ is constructed with the help of the Fourier transform (see \eqref{B.3}). For Schr\"odinger operators with short range potentials, $\Gamma(\lambda)$ is closely related to the generalized eigenfunctions constructed through a Lippmann-Schwinger type argument starting from the free plane-waves (see \eqref{B8}). 

Then for every function $\phi\in C_0^\infty(J_a)$ and for every $f\in \cD$ we have:
\begin{equation}\label{Gammaop'}
\langle f, \phi(H)f  \rangle_{\cH} = \phi(0)|\langle \Psi_0,f\rangle|^2+\int_{J_a}\phi(\lambda)\langle \Gamma(\lambda)f, \Gamma(\lambda)f \rangle_{\sfh} d\lambda.
\end{equation}
By a standard limiting argument this implies that for every $J\subseteq J_a$ and for every $f\in\cD$ we have
\begin{equation}\label{Gammaop}
||Q(J)f||^2=\langle f, Q(J)f  \rangle = \int_{J}\langle \Gamma(\lambda)f, \Gamma(\lambda)f \rangle_{\sfh} d\lambda.
\end{equation}
The next step is to define 
\begin{equation}\label{Gammachar}
J_a\ni\lambda\mapsto (\widetilde{\Gamma} f)(\lambda):=\Gamma(\lambda)f \in \sfh,\quad f\in\cD,
\end{equation}
which due to \eqref{Gammaop}  can be extended by continuity to a partial  isometry between $\cH$ and  $L^2(J_a,\sfh)$ such that 
$$\int_{J_a}||(\widetilde{\Gamma} \psi)(\lambda)||^2_{\sfh}d\lambda =||Q(J_a)\psi||^2_{\cH}.$$

From now on we will make the following assumption.
\begin{assumption}\label{SM}
We have that $A^{\ast}\cK\subseteq\cD$ and furthermore:
\begin{equation}\label{assum29}
\Gamma(\cdot) A^*\in C^{2,\omega_{\theta}}(J_a;\cB(\cK,\sfh)). 
\end{equation}
\end{assumption}

For $\im z \neq 0$ we define
\begin{align}\label{Sz}
Q_0(H-z)^{-1}Q_0&=Q(J_a)(H-z)^{-1}Q(J_a)\notag\\
&\quad+P(\bR \setminus J_a)(H-z)^{-1}P(\bR \setminus J_a)\notag\\
&=:S_<(z)+S_>(z)=:S(z),
\end{align} 
where these operators act on $\cH$. Notice that $S_>(z)$ is analytic in $|z|<a$.
For every $g\in\cH$ we have
\begin{align}\label{SzGnorm}
S_<(z)g&=\widetilde{\Gamma}^*\Bigl(\frac{1}{\cdot -z}
(\widetilde{\Gamma}g)(\cdot)\Bigr),\notag\\
||S_<(z)g||^2&=\int_{J_a} \frac{1}{|\lambda-z|^2}\|(\widetilde{\Gamma}g)(\lambda)\|_{\sfh}^2d\lambda.
\end{align}

The operator
\begin{equation}
S_<:=Q(J_a)H^{-1}Q(J_a),
\end{equation}
is well defined on the domain:
\begin{equation}\label{DS<}
\cD(S_<):=\bigl\{g\in \cH:\; \int_{J_a}\frac{\| (\widetilde{\Gamma}g)(\lambda)
\|^2_{\sfh}}{\lambda^2}d\lambda < \infty \bigr\}.
\end{equation}

The operator $S_<$ is self-adjoint and unbounded because $0$ belongs to the continuous spectrum of $H$. For any $g\in \cD(S_<)$ we have:
\begin{equation}\label{S<}
S_<g=\widetilde{\Gamma}^*\Bigl(\frac{1}{\cdot}
(\widetilde{\Gamma}g)(\cdot)\Bigr).
\end{equation} 

We now define the operator:
\begin{equation}\label{S}
S:= S_< +S_>(0),
\end{equation}
which is self-adjoint on $\cD(S_<)$. 

Due to \eqref{fact} and Assumption \ref{SM} we know that $\Gamma(\cdot) WP_0\in C^{2,\omega_{\theta}}(J_a;\cB(\cH,\sfh))$ for some $\theta \in (0,1)$. Thus  the Fermi Golden Rule constant 
(see\eqref{GammaFRG}) reads as 
$\Gamma_{FGR}=\pi \|\Gamma(0)W\Psi_0\|_{\sfh}^2$, hence the assumption that $\Gamma_{FGR}=0$ is equivalent with:

\begin{assumption}\label{Gamma(0)}
We have  
$\Gamma(0)WP_0=0$
as an operator from $\cH$ to $\sfh$.
\end{assumption}

Excluding the trivial case when $E_0$ is
isolated, the operator $S$ is self-adjoint but unbounded, i.e. $SWP_0$ is not bounded 
in general, and this is the reason for the breakdown of the R-S expansion in the embedded case. Our first 
simple but important result says that  
 $SWP_0$ remains bounded if   Assumptions~\ref{SM} and~\ref{Gamma(0)} hold true.

\begin{proposition}\label{bdd}
Suppose Assumptions~\ref{SM}  holds true.
Then $ WP_0\cH \subseteq \cD(S_<)$ if and only if Assumption~\ref{Gamma(0)} holds true.
In particular if Assumptions~\ref{SM} and~\ref{Gamma(0)} hold true, then $SWP_0$ is bounded and its adjoint is the 
extension by continuity  of $P_0WS$.
\end{proposition}

Now consider:
\begin{equation}\label{T1}
T_1:=-SWP_0-P_0WS,
\end{equation}
which in regular perturbation theory gives the first order correction
in the expansion of the perturbed eigenprojection, and
\begin{equation}\label{Tve}
T_{\ve}:=P_0+\ve T_1.
\end{equation}
By a simple computation we obtain:
\begin{equation}\label{Deltave}
 T_{\ve}^2-T_{\ve}=\ve^2T_1^2,
\end{equation}
thus $T_{\ve}$ is an `almost orthogonal projection'. There exists a whole family of orthogonal projections which are
close to $T_{\ve}$ in norm, up to errors of order $\ve^2$. The following one is
distinguished by the fact that it is given by the 
following 
algebraic formula, assuming $\| T_{\ve}^2-T_{\ve} \| <\frac{1}{4} $, see \cite{N02} and \cite{M}:
\begin{equation}\label{Pve}
P_{\ve}=T_{\ve}+(T_{\ve}-1/2)[(1+4(T_{\ve}^2-T_{\ve}))^{-\frac{1}{2}}-1].
\end{equation}
Notice that $P_{\ve}-P_0 =\cO(\ve)$, $P_{\ve}-T_{\ve}=\cO(\ve^2)$. Let
now $U_{\ve}$ be the Sz.-Nagy unitary intertwining $P_{\ve}$ 
and $P_0$:
\begin{equation}\label{Uve}
U_{\ve}=\frac{1}{(1-(P_{\ve}-P_0)^2)^{\frac{1}{2}}}(P_{\ve}P_0 +(1-P_{\ve})(1-P_0)),\quad P_{\ve}=U_{\ve}P_0U_{\ve}^*.
\end{equation}	
To establish the improved exponential decay law we take as our ansatz for the initial state:
\begin{equation}\label{Psi1ve}
\Psi^1_{\ve}:=U_{\ve}\Psi_0= \Psi_0 -\ve SW\Psi_0 +\cO(\ve^2).
\end{equation}
In other words, we now have to estimate: 
\begin{equation}\label{A1}
A^1(\ve,t):=\langle \Psi^1_{\ve},e^{-itH_{\ve}}\Psi^1_{\ve}\rangle.
\end{equation}
A remark is in order here. One can equally well use the following (simpler at  first sight) 
ansatz: $\widetilde{\Psi}^1_{\ve}=T_{\ve}\Psi_0/\| T_{\ve}\Psi_0 \|$. The reasons for choosing \eqref{Psi1ve} 
are that the proofs are somewhat simpler, and more importantly, the procedure extends unchanged to the degenerate case.

The aim in what follows is to find $E^1_{\ve}$ such that:
\begin{equation}
|A^1(\ve,t)-e^{-itE^1_{\ve}}|\lesssim \ve^4.
\end{equation}
Using the Stone and SLFG formulae one obtains as in \cite{JN2}:
\begin{equation}\label{Main1}
A^1(\ve,t)=
\lim_{\eta \searrow 0}\frac{1}{2\pi i}\int_{\bR}dx\, e^{-ixt}
\Bigl(\frac{1}{F^1(x+i\eta,\ve)}-\frac{1}{F^1(x-i\eta,\ve)}\Bigr),
\end{equation}
where
\begin{align}\label{Fone}
F^1(z,\ve)&:= \langle\Psi^1_{\ve},H_{\ve}\Psi^1_{\ve}\rangle-z \notag\\
&\qquad-\langle\Psi^1_{\ve},P_{\ve}H_{\ve}Q_{\ve}(Q_{\ve}H_{\ve}Q_{\ve}-z)^{-1}Q_{\ve}H_{\ve}P_{\ve}\Psi^1_{\ve}\rangle
\end{align}
and $Q_{\ve}=1-P_{\ve}$.
In order to follow the line of the proof of Theorem \ref{JN06} in \cite{JN2}, it is convenient to use \eqref{Uve}, 
\eqref{Psi1ve}, and rewrite \eqref{Fone} as
\begin{equation}\label{F^1}
F^1(z,\ve)= \langle\Psi_{0},\widetilde{H}_{\ve}\Psi_{0}\rangle-z
-\langle\Psi_{0},P_{0}\widetilde{H}_{\ve}Q_{0}(Q_{0}\widetilde{H}_{\ve}Q_{0}-z)^{-1}Q_{0}\widetilde{H}_{\ve}P_{0}\Psi_{0}\rangle,
\end{equation}
where
\begin{equation}\label{tHve}
\widetilde{H}_{\ve}=U_{\ve}^*H_{\ve}U_{\ve}.
\end{equation}
Using the notation
\begin{equation}\label{tW}
\widetilde{W}_{\ve}:=\frac{1}{\ve}(\widetilde{H}_{\ve}-H),
\end{equation}
the formula \eqref{F^1} becomes
\begin{align}\label{F^11}
F^1(z,\ve)&=\ve \langle\Psi_{0},\widetilde{W}_{\ve}\Psi_0\rangle-z\notag\\
&\qquad-\ve^2\langle\Psi_{0},P_{0}\widetilde{W}_{\ve}
Q_{0}(Q_{0}(H+\ve\widetilde{W}_{\ve})Q_{0}-z)^{-1}Q_{0}
\widetilde{W}_{\ve}P_{0}\Psi_{0}\rangle.
\end{align}

The following lemma lists some of the important structural properties of $\widetilde{W}_{\ve}$. 

\begin{lemma}\label{W} 
The following results hold for $\widetilde{W}_{\ve}$.
\begin{itemize}
\item[\rm(i)] 
\begin{align}\label{tildeWve3}
\widetilde{W}_{\ve}&= W-P_0WQ_0-Q_0WP_0\notag\\
&\quad+\ve\bigl(-P_0WSW-WSWP_0+SWP_0W+WP_0WS\notag\\
&\quad+P_0WSWP_0-\frac{1}{2}Q_0WP_0WS -\frac{1}{2}SWP_0WQ_0\bigr) + \cO(\ve^2).
\end{align}
In particular
\begin{equation}\label{P0Q0}
D_{\ve}:=\frac{1}{\ve} P_0\widetilde{W}_{\ve}Q_0= -P_0WSWQ_0 -P_0WP_0WS + \cO(\ve)
\end{equation}
is uniformly bounded as $\ve \rightarrow 0$.

\item[\rm(ii)] There exist some operators $  V_{j k}(\ve)$, $1\leq j,k\leq
4$ which are uniformly bounded as $\ve \rightarrow 0$, such that if $\mathcal{V}(\ve):=\begin{bmatrix}
  V_{j k}(\ve)
 \end{bmatrix}$ denotes the obvious operator valued matrix, we have:
\begin{equation}\label{matrixW}
\widetilde{W}_{\ve}=
 \begin{bmatrix}
  SWP_0 &Q_0A^* &P_0WS & P_0
 \end{bmatrix}
\mathcal{V}(\ve)
\begin{bmatrix}
 P_0WS\\
 AQ_0\\
 SWP_0\\
 P_0\\
\end{bmatrix}.
\end{equation}
If we denote by $X\colon\mathcal{H}\to \cH\oplus\cK\oplus\cH\oplus\cH$ the map 
$$X(f):=\begin{bmatrix}
 P_0WSf\\
 AQ_0f\\
 SWP_0f\\
 P_0f\\
\end{bmatrix},$$
then we can write:
\begin{equation}\label{horia1}
\widetilde{W}_{\ve}=X^*\mathcal{V}(\ve)X.
\end{equation}
\end{itemize}
\end{lemma}

Using \eqref{P0Q0}, we can express $F^1(z,\ve)$ as:
\begin{align}\label{F^22}
F^1(z,\ve)&=\ve \langle\Psi_{0},\widetilde{W}_{\ve}\Psi_0\rangle-z\notag\\
&\quad-\ve^4\langle\Psi_{0},P_{0}
D_{\ve}Q_{0}(Q_{0}(H+\ve\widetilde{W}_{\ve})Q_{0}-z)^{-1}Q_{0}
D_{\ve}^*P_{0}\Psi_{0}\rangle. 
\end{align}

The next Lemma, whose proof is based on a standard perturbation theory
argument, is a direct consequence of Lemma~\ref{W} and reduces the
smoothness problem of the r.h.s. of \eqref{F^22} to the smoothness of four $\ve$-independent expressions. 
In this Lemma and the next $\cB$ is a shorthand for the four spaces $\cB(\cH)$, $\cB(\cK)$, $\cB(\cH,\cK)$, and $\cB(\cK,\cH)$, as appropriate.
\begin{lemma}\label{P}
Suppose that the following four families of operators 
\begin{align}\label{smth} 
&P_0WSS(z)SWP_0, &&AQ_0S(z)SWP_0,&\notag\\
&P_0WSS(z)Q_0A^*, &&AQ_0S(z)Q_0A^*,& 
\end{align}
belong to $ C^{1,\omega}(D_{a/2};\cB)$. Then: 
\begin{equation}
 \langle\Psi_{0},P_{0} D_{\ve}Q_{0}(Q_{0}(H+\ve\widetilde{W}_{\ve})Q_{0}-z)^{-1}Q_{0}
 D_{\ve}^*P_{0}\Psi_{0}\rangle \in C^{1,\omega}(D_{a/2};\bC).
\end{equation}
\end{lemma}
The next lemma shows that the hypothesis of Lemma~\ref{P} are indeed satisfied in the case when $\omega=\omega_\theta$:
\begin{lemma}\label{smooth}
Suppose Assumptions~\ref{SM} and~\ref{Gamma(0)} hold true. Then the
four families of Lemma \ref{P} belong to $C^{1,\omega_{\theta}}(D_{a/2};\cB)$.
\end{lemma}
A consequence of Lemmas~\ref{P} and~\ref{smooth} is that the map
$F^1(\cdot,\ve)$ lies in $C^{1,\omega_{\theta}}(D_{a/2})$. Now we can formulate the 
main result of the paper, showing that  the replacement of $\Psi_0$ with $\Psi^1_{\ve}$ leads to 
an improvement of the error term in the exponential decay law, namely the error term is of order at most $\ve^4$.
\begin{theorem}\label{ve4}
Suppose Assumptions~\ref{SM} and~\ref{Gamma(0)} hold true. 
Then for sufficiently small $\ve$ we have
$$F^1(\cdot,\ve) \in C^{1,\omega_{\theta}}(D_{a/2};\bC).$$
In particular, it has well defined limit values
$$ F^1(x,\ve) := \lim_{\eta \searrow 0}F^1(x + i\eta,\ve) \in C^{1,\omega_{\theta}}(J_{a/2};\bC).$$
Let $R^1(x,\ve)$ and $I^1(x, \ve)$ be the real and imaginary part of $F^1(x,\ve)$, respectively:
\begin{equation}
F^1(x,\ve) =: R^1(x,\ve)+ i I^1(x, \ve).
\end{equation}
For a fixed $\ve$ the equation
\begin{equation}\label{x1}
R^1(x,\ve)=0
\end{equation}
has a unique solution $x^1(\ve)$ in $J_{a/2}$, with $|x^1(\ve)| \lesssim \ve$ .
Define
\begin{equation}\label{E1ve}
E^1_{\ve}:=x^1(\ve)+iI^1(x^1(\ve),\ve).
\end{equation}
Then for sufficiently small $\ve$ we have
\begin{equation}\label{A1main}
|A^1(\ve, t)-e^{-itE^1_{\ve}}| \lesssim \ve^4.
\end{equation}
\end{theorem}

The proofs of both theorems rely heavily  on a careful estimate of
the integrals in the r.h.s. of \eqref{Main1} and \eqref{Main0}, respectively. 
The following technical lemma provides an abstract setting which can
be  used directly in both theorems (note that $\int_0^1\frac{\omega_{\theta}}{x}dx <\infty$ for all $\theta >0$).

\begin{lemma}\label{Techest}
Consider the function
\begin{equation}\label{Fz}
F(z,\ve)= a(\ve)-z -\gamma(\ve)f(z,\ve),
\end{equation}
for which the following five conditions hold true:
\begin{itemize}
\item[\rm(i)] $f(z,\ve)$ is an analytic function on $\{z\in\bC: \im z>0\}$, and  $\overline{f(\overline{z},\ve)}=f(z,\ve)$;
\item[\rm(ii)] Let $a(\ve)$ be real valued, $\gamma(\ve)>0$  and $\lim_{\ve
  \rightarrow 0} (|a(\ve)| +\gamma(\ve))=0$;
\item[\rm(iii)] $\im f(z,\ve) \geq 0$ for $\im z >0$;
\item[\rm(iv)] $\lim_{\eta \searrow 0}\frac{1}{\pi} \int_{\bR}\frac{\im F(x +i\eta, \ve)}{|F(x +i\eta, \ve)|^2}dx=1$;
\item[\rm(v)]  $f(\cdot , \ve) \in C^{1,\omega}(D_b;\bC)$ for some $b>0$
and $\int_0^1\frac{\omega(x)}{x}dx <\infty$.
\end{itemize}
Let $R(x,\ve)$, $I(x, \ve)$ be the real and imaginary part of $\lim_{\eta \searrow 0}F(x + i\eta,\ve)$ for $x\in J_b$.
Then the following statements hold true:
\begin{itemize}
\item[\rm(a)] For a sufficiently small $\ve$, the equation
\begin{equation}\label{x}
R(x,\ve)=0
\end{equation}
has a unique solution $x(\ve)\in J_b$, obeying: 
\begin{equation}\label{xve}
|x(\ve)| \lesssim |a(\ve)|+ \gamma(\ve).
\end{equation}
\item[\rm(b)]  If 
\begin{equation}\label{Eve}
E_{\ve}:=x(\ve)+iI(x(\ve),\ve),
\end{equation}
then for sufficiently small $\ve$ we have:
\begin{align}\label{Main}
\limsup_{\eta \searrow 0}\bigl|\frac{1}{2\pi i}\int_{\bR} e^{-ixt}
&\Bigl(\frac{1}{F(x+i\eta,\ve)}\notag\\
&\qquad-\frac{1}{F(x-i\eta,\ve)}\Bigr)dx-e^{-itE_{\ve}}
\bigr| \lesssim \gamma(\ve).
\end{align}
\end{itemize}
\end{lemma}

In both Theorem~\ref{JN06} and Theorem~\ref{ve4} 
two functions $F^0(z,\ve)$ and $F^1(z,\ve)$ appear, and
the above lemma has to be applied to each of them. 

It is important to remember
that up to an application of the SLFG formula, we have that
$1/F^0(z,\ve)=\langle \Psi_0, (H+\ve W-z)^{-1}\Psi_0\rangle$ and 
$1/F^1(z,\ve)=\langle \Psi_0, (H+\ve
\widetilde{W}_{\ve}-z)^{-1}\Psi_0\rangle$. 
In this case, the limit $\lim_{\eta \searrow 0}\frac{1}{2\pi i}\int_{\bR} e^{-ixt}
(\frac{1}{F(x+i\eta,\ve)}-\frac{1}{F(x-i\eta,\ve)})dx$ exists due to the Stone formula (see~\eqref{Main0} and~\eqref{Main1}). Moreover, the same Stone formula
shows that the condition (iv) in the
lemma is nothing but the fact that when $t=0$, the evolution group
equals the identity, hence $A^0(\ve,0)=A^1(\ve,0)=1$. 

We will see that in the case of Theorem \ref{JN06} we have  
$\gamma(\ve) =\ve^2$, while 
in the case of Theorem \ref{ve4} we have $\gamma(\ve) =\ve^4$. 
An outline of the proof of Lemma \ref{Techest} for the case $\omega(x)=x^{\theta}$ has been given in \cite{JN2}. 
For completeness, we shall provide the proof in the next section.

There are many open questions in this area and we close this section by mentioning a few of  them.
\begin{itemize}
\item[\rm(i)] The condition $\Gamma_{FGR}=0$ is crucial in proving Theorem~\ref{ve4}. A natural question is whether
 in case $\Gamma_{FGR}\neq 0$ one can still find a better resonance function
 than $\Psi_0$ (e.g. by truncating the
 corresponding Gamov vector~\cite{Sk})  leading to an error term of order $\ve^q$, $q>2$. 
 We believe that this is not possible.

\item[\rm(ii)] Theorem \ref{ve4} gives $\im E^1_{\ve} = -\ve^4 \Gamma_{FGR}^1
+\cO(\ve^5)$ and moreover, $\Gamma_{FGR}^1$ can be  computed explicitly. 
At the formal level, one can show that via the result in \cite{Ho1},
$\Gamma_{FGR}=\Gamma_{FGR}^1=0$ insures that the R-S 
expansion is well 
defined up to $N=2$ so that the procedure can be iterated, leading to an
even smaller error term. Unfortunately, the proof of the 
 required smoothness 
properties becomes prohibitively complex.  

Indeed, notice that while for $F^0$ the required smoothness
properties follow directly from the smoothness of $Q_0(H-z)^{-1}Q_0$, in the case of $F^1$ the proof of the required smoothness properties constitutes
most of the technicalities of this paper, see Lemmae~\ref{W}--\ref{smooth}. Moreover, the proof of some crucial properties of 
 $\widetilde{W}_{\ve}$ listed in Lemma~\ref{W} are long and somewhat tricky explicit computations. While there is no doubt that with a 
large amount
 of computations and technicalities one can prove the analogue of Theorem~\ref{ve4} for $N=2$, for $N>2$ one has to invent a clever
recurrence  scheme. If that can be achieved, the outcome will be a fully developed R-S perturbation theory, as well as a control of dynamics of
the corresponding metastable states for embedded
 eigenvalues in the smooth context.

\item[\rm(iii)] There is a large body of literature on resonances of Liouvilleans governing the dynamics of open quantum systems,
see e.g \cite{Me} and references given there. In particular, it has been shown in \cite{Me} that if the Hamiltonian describing the 
small component of the open quantum system has degenerate eigenvalues, then the corresponding FGR constant vanishes, and this  leads to
return to equilibrium with exponentially fast convergence of order $e^{-\xi\lambda^4t}$, where $\xi>0$ and $\lambda$ is the 
strength of  the coupling between the small system and the reservoir. It would be very interesting to see whether some of the 
techniques developed in the present paper could be adapted to the Liouvillean framework.

\end{itemize}

\section{Proofs}\label{section3}
\subsection { Proof of Proposition \ref{lowb}}
Let $P_{\ve}$ be the orthogonal projection corresponding to
$\Psi_{\ve}$. Since
$\langle\Psi_0,\Psi_1\rangle=0$, we have $P_0\Psi_1=0$ and
$Q_0\Psi_1=\Psi_1$.  

From \eqref{Psi0Psi} we get:
$$A
^0(\ve,t)-e^{-itE_{\ve}}=\ve^2\langle\Psi_1,(e^{-itH_{\ve}}-e^{-itE_{\ve}})\Psi_1 \rangle +o(\ve^2).
$$
Let $Q_{\ve}=1-P_{\ve}$. We have that 
$\|P_{\ve}-P_0\|=\cO(\ve)$, $\|P_{\ve}Q_0\|=\cO(\ve)$ and $\|P_{\ve}\Psi_1\| =\cO(\ve)$. Then
$$
\langle\Psi_1,(e^{-itH_{\ve}}-e^{-itE_{\ve}})\Psi_1 \rangle=
\langle Q_{\ve}\Psi_1, e^{-itH_{\ve}} Q_{\ve}\Psi_1\rangle -e^{-itE_{\ve}}\|\Psi_1\|^2 +\cO(\ve).
$$
Now, as $Q_{\ve}\cH$ maps into the subspace of absolute continuity for $H_{\ve}$
the first term in the r.h.s. vanishes in the limit $t \rightarrow
\infty$. Summing up:
$$
\liminf_{t \rightarrow \infty}| \langle\Psi_1,(e^{-itH_{\ve}}-e^{-itE_{\ve}})\Psi_1 \rangle| = \|\Psi_1\|^2 +\cO(\ve)
$$
and the proof is finished as $\|\Psi_1\| \neq 0$.

\subsection{ Proof of Proposition \ref{bdd}} 
If $g \in  \cH$ with $\|g\|=1$, then 
$h(\lambda):=(\widetilde{\Gamma} WP_0g)(\lambda)\in \sfh$. We see that $h(\lambda)=\Gamma(\lambda)WP_0g$ because of \eqref{Gammachar} and Assumption \ref{SM}. In addition 
\begin{equation}\label{Lip}
h(0)=\Gamma(0)W\Psi_0 \langle \Psi_0,g\rangle,\quad \| h(\lambda)-h(0)\|_{\sfh} \leq  C|\lambda|,\quad \lambda\in J_a,
\end{equation}
where $C <\infty$ is a constant which is independent of $g$. From \eqref{DS<} and \eqref{S<} we see that $WP_0g$ belongs to the domain of $S_<$ if and only if 
$$\int_{J_a}\frac{\|h(\lambda)\|^2_{ \sfh}}{\lambda^2} d\lambda <\infty.
$$
From \eqref{Lip} it follows that if  Assumption \ref{Gamma(0)} also holds true, then $h(0)=0$ and $S_<WP_0$ is bounded. If Assumption \ref{Gamma(0)} does not hold true, then $W\Psi_0$ does not belong to the domain of $S_<$.

Now let us assume that both Assumptions  \ref{SM} and \ref{Gamma(0)} hold. 
Since $S_>WP_0$ is also bounded, it follows that $SWP_0$ is
bounded. On the other hand, for $f\in \cD(S_<)$, we have by duality
\begin{align*}
 \|P_0WS_<f\| &= \sup_{\|g\|=1}|\langle S_<WP_0g,f\rangle| \\
&\leq (\sup_{\|g\|=1}\|S_<WP_0g\|)\|f\|\leq \|S_<WP_0\| \|f\|,
\end{align*}
and the proof is finished.

\subsection{Proof of Lemma \ref{W}}

The proof requires tedious computations using \eqref{tHve}, \eqref{tW}, \eqref{Pve}, and \eqref{Uve}. We have to prove that $\widetilde{W}_{\ve}$ 
is uniformly bounded as 
$\ve \rightarrow 0$, even though $H$ is not supposed to be
bounded. We will show that $\widetilde{W}_{\ve}$ has a norm expansion in $\ve$ and compute this expansion up to errors of order
$\cO(\ve^3)$.

To proceed we need to expand the term
 \begin{equation}\label{Dve}
 \Delta_{\ve}:=P_{\ve}-P_0. 
 \end{equation}
 From now 
on we shall take $\ve$ sufficiently small such that Proposition~\ref{bdd}, \eqref{T1},
 \eqref{Tve}, and \eqref{Deltave} imply
\begin{equation}\label{ve<1}
 \|T_{\ve}^2-T_{\ve}\| =\ve^2\|T_1\|^2<\tfrac{1}{4},
\end{equation}
and then one can expand in powers of $T_{\ve}^2-T_{\ve}$ in the r.h.s. of \eqref{Pve} and obtain
\begin{equation}\label{Pve3}
 \Delta_{\ve}= \ve T_1 +\ve^2\bigl(1-2P_0\bigr)T_1^2 +\ve^3 \bigl(P_0 E_0(\ve) +T_1E_1(\ve)\bigr)T_1
\end{equation}
with $E_j(\ve)$ uniformly bounded as 
$\ve \rightarrow 0$.

The following identities  follow from the fact that $P_{\ve}$ and $P_0$ are projections.
\begin{equation}\label{id1}
[ \Delta_{\ve}^2,P_{\ve}]=[ \Delta_{\ve}^2,P_{0}]=0,
\end{equation}
\begin{equation}\label{id2}
 P_{\ve}P_0+(1-P_{\ve})(1-P_0)=1-(2P_{\ve}-1)\Delta_{\ve}=1+\Delta_{\ve}(2P_0-1).
\end{equation}
For $\ve$ sufficiently small such that in addition to
\eqref{ve<1} we have 
$\|\Delta_{\ve}\|<1$. 
Then
\begin{equation}\label{expD}
 (1-\Delta_{\ve}^2)^{-\frac{1}{2}}=:
 1+\Delta_{\ve}^2N_{\ve}=:
1+\tfrac{1}{2}\Delta_{\ve}^2 +\Delta_{\ve}^2\widetilde{N}_{\ve}\Delta_{\ve}^2,
\end{equation}
with $N_{\ve}$, $\widetilde{N}_\ve$ uniformly bounded as $\ve \rightarrow
0$, and commuting with both $P_0$ and $P_{\ve}$. 
Inserting \eqref{Dve} in \eqref{Uve} and using \eqref{expD}, \eqref{id1},   \eqref{id2} one obtains
\begin{equation}\label{expU}
U_{\ve}=1+\tfrac{1}{2}\Delta_{\ve}^2+\Delta_{\ve}(2P_0-1)+\Delta_{\ve}^2\widetilde{N}_{\ve}\Delta_{\ve}^2
-\Delta_{\ve}^2N_{\ve}(2P_{\ve}-1)\Delta_{\ve}.
\end{equation}
 Then we write
\begin{equation}\label{Bve}
 B_{\ve}:=U_{\ve}-1=:\Delta_{\ve}(2P_0-1)+\tfrac{1}{2}\Delta_{\ve}^2 +\Delta_{\ve}^2M_{\ve}\Delta_{\ve}
\end{equation}
with $M_{\ve}$ uniformly bounded as $\ve \rightarrow 0$. 

By direct computation we obtain (see \eqref{tHve} and \eqref{tW}):
\begin{align}\label{Wve}
 \widetilde{W}_{\ve} = B_{\ve}^*H+HB_{\ve} +B_{\ve}^*HB_{\ve} +\ve\bigl(W+B_{\ve}^*W +WB_{\ve} +B_{\ve}^*WB_{\ve}\bigr).
\end{align}
To proceed further, we list the following identities which follow from $E_0=0$ and \eqref{T1}.
\begin{equation}\label{id3}
 HP_0=P_0H=0,\;\;\; HT_1=-Q_0WP_0, \;\;\; T_1Q_0=P_0T_1=-P_0WS.
\end{equation}
From \eqref{Pve3} and \eqref{id3} we get
\begin{equation}\label{HPve}
H\Delta_{\ve}= HP_{\ve}=-\ve Q_0WP_0+\ve^2Q_0WP_0WS -\ve^3 Q_0WP_0E_1(\ve)T_1,
\end{equation}
\begin{equation}\label{PveQ0}
\Delta_{\ve}Q_0=-\ve P_0WS+\ve^2\bigl(1-2P_0\bigr)T_1^2Q_0 -\ve^3 \bigl(P_0E_0(\ve)+T_1E_1(\ve)\bigr)P_0WS.
\end{equation}
From this point onwards the proof of Lemma \ref{W} is a somewhat long but straightforward computation using \eqref{id3}. 
Consider for example the term $HB_{\ve}$. From \eqref{Bve} and \eqref{HPve}
we get
 \begin{multline}
 HB_{\ve}=\bigl(-\ve Q_0WP_0+\ve^2Q_0WP_0WS -\ve^3 Q_0WP_0E_1(\ve)T_1\bigr)
 \\
\cdot\bigl(2P_0-1 +\tfrac{1}{2}\Delta_{\ve}+ \Delta_{\ve}M_{\ve}\Delta_{\ve}\bigr),
 \end{multline}
and using  \eqref{Pve3} one can see by inspection that it has the structure in \eqref{matrixW} and also compute
 explicitly its  expansion up to terms of order $\cO(\ve^2)$. 

\subsection{ Proof of Lemma \ref{P}}

Take $\im z> 0$. We want to write $\bigl(Q_0\widetilde{H}_{\ve}Q_0-zQ_0\bigr)^{-1}$
in the Hilbert space $Q_0\cH$ in a different way. Using
\eqref{horia1}, by a standard re-summation argument we obtain
\begin{align}\label{horia2}
 \bigl(Q_0\widetilde{H}_{\ve}Q_0-zQ_0\bigr)^{-1}&=\bigl(Q_0HQ_0-zQ_0\bigr)^{-1}
 -\ve \bigl(Q_0HQ_0-zQ_0\bigr)^{-1}Q_0X^*\notag\\
&\quad\cdot \bigl\{1+\ve \mathcal{V}_{\ve} XQ_0
(Q_0HQ_0-zQ_0)^{-1}Q_0X^*\bigr\}^{-1}\notag\\
&\quad\cdot\mathcal{V}_{\ve}XQ_0\bigl(Q_0HQ_0-zQ_0\bigr)^{-1}.
\end{align}
 Note that both $X^*$ and $X$
contain localizing factors. If $z\in D_a$ and $\ve$ is
small enough, then due to our hypothesis on the four families of
operators we have the uniform bound
$$\ve \| \mathcal{V}_{\ve} XQ_0
  \bigl(Q_0HQ_0-zQ_0\bigr)^{-1}Q_0X^*\|\leq \tfrac12,$$
hence the representation \eqref{horia2} makes sense in
$D_a$ up to the real line. Now by standard resolvent identities we can transfer
the smoothness of $XQ_0
  (Q_0HQ_0-zQ_0)^{-1}Q_0X^*$ to the inverse
$$ \bigl\{ 1+\ve \mathcal{V}_{\ve}XQ_0
  (Q_0HQ_0-zQ_0)^{-1}Q_0X^*\bigr\}^{-1}.$$
Also, from \eqref{HPve}, \eqref{Wve}, \eqref{Bve}, and \eqref{Pve3} we
conclude that $P_{0} D_{\ve}Q_{0}$ contains localizing factors,
uniformly in $\ve$. Using again \eqref{horia2}, we obtain:
$$
\langle\Psi_{0},P_{0} D_{\ve}Q_{0}\bigl(Q_{0}(H+\widetilde{W}_{\ve})Q_{0}-zQ_0\bigr)^{-1}Q_{0} D_{\ve}^*P_{0}\Psi_{0}\rangle \in
 C^{1,\omega_{\theta}}(D_{a/2};\bC).
 $$

\subsection{Proof of Lemma \ref{smooth}}

It is sufficient to consider the operators in \eqref{smth} with
$S$ and $S(z)$ replaced by $S_<$ and $S_<(z)$, respectively. We only consider 
$P_0WS_<S_<(z)Q_0A^*$, the others can be treated similarly.  From  \eqref{SzGnorm}, \eqref{S<}, \eqref{Gammachar} and the fact that $A^*$ maps $\cK$ into $\cD$ we have:
\begin{equation*}
 \langle f,P_0WS_<S_<(z)Q_0A^*g\rangle=\int_{J_a}\frac{1}{\lambda(\lambda-z)} \langle \Gamma(\lambda)WP_0f,\Gamma(\lambda)A^*g\rangle_{\sfh} d\lambda.
\end{equation*}
Let us consider (we also use Assumption \ref{Gamma(0)}):
\begin{align*}
\Phi(\lambda)&=\frac{1}{\lambda}\bigl\langle [\Gamma(\lambda)-\Gamma(0)]WP_0f,\Gamma(\lambda)A^*g\bigr \rangle_{\sfh}\\
&=\int_0^1\left \langle [\Gamma'(u\lambda)WP_0]f,\Gamma(\lambda)A^*g\right \rangle_{\sfh} du.
\end{align*}
Thus there exists a constant $C>0$ such that 
\begin{equation}\label{smf}
 \max\{|\Phi(\lambda_1)-\Phi(\lambda_2)|, |\Phi'(\lambda_1)-\Phi'(\lambda_2)|\}\leq C \|f\|\|g\|\; \omega_{\theta}(|\lambda_1-\lambda_2|), 
\end{equation}
for all $\lambda_1,\lambda_2\in J_a$.
Let  $\chi \in C^{\infty}({J_a})$, $0 \leq \chi(\lambda) \leq 1$,  $\chi(\lambda)=1$ for $|\lambda| <\frac{3}{4}a$,
  $\chi(\lambda)=0$ for $|\lambda| >\frac{7}{8}a$, and write
\begin{multline}\label{chif}
\langle f,P_0WS_<S_<(z)Q_0A^*g\rangle=\int_{J_a}\frac{\Phi(\lambda)}{\lambda-z} d\lambda \\
=\int_{J_a}(1-\chi(\lambda))\frac{\Phi(\lambda)}{\lambda-z} d\lambda +
\int_{J_a}\frac{\Phi(\lambda)\chi(\lambda)}{\lambda-z} d\lambda.
\end{multline}
The first term in the r.h.s. of \eqref{chif} is analytic in $|z| < \frac{3}{4}a$, while for the second one we use that the Cauchy integral transform a compactly supported 
function from $C^{n, \omega_{\theta}}(J_{a};\bC)$ is continuously mapped to $C^{n, \omega_{\theta}}(D_{a/2};\bC)$ (see Appendix A). Due to \eqref{smf} we can lift the weak estimate to a norm estimate, and the 
proof is over. 

\subsection{ Proof of Lemma \ref{Techest}}

Write for $z=x+i\eta \in D_b$:
\begin{align} \label{F}
 F(x+i\eta, \ve) &= \re  F(x+i\eta, \ve) +i \im  F(x+i\eta, \ve)\notag\\
&=: R(x, \ve, \eta)+ iI(x, \ve, \eta).
\end{align}
It follows from \eqref{Fz} and v. that for $x\in J_b$
\begin{equation}
F(x, \ve):= \lim_{\eta \searrow 0}F(x + i\eta,\ve)
\end{equation}
exists and belongs to $ C^{1,\omega}(J_b;\bC)$.
It follows from \eqref{Fz}, ii. and v. (notice that $\frac{d}{dx}R(x, \ve, \eta)$ is close to $-1$ and
 $F(\cdot , \ve) \in C^{1,\omega}(D_b;\bC)$). Thus 
for $|x|<b/2$ and $\ve, \eta \geq 0$ sufficiently small the equation $
 R(x, \ve, \eta)=0$ 
has  a unique solution $x(\ve, \eta)$ in $|x| <b/2$, $|x(\ve, \eta)| \lesssim |a(\ve)|+ \gamma(\ve)$. In addition, $
 \lim_{\eta \searrow 0}x(\ve, \eta)=x(\ve)$. 
It follows from condition (iii) that $I(x, \ve, \eta) <0$.
Notice also that on $D_b$
\begin{equation}\label{<I<}
0<-I(x, \ve, \eta)\lesssim \eta +\gamma(\ve).
\end{equation}
 Now fix $\ve$ sufficiently small. Due to \eqref{<I<} one can find $C>0$ such that for $\eta$ sufficiently small
\begin{equation}\label{interval}
J_{\ve,\eta}=\bigl[x(\ve, \eta) - C \frac{\Gamma (\ve, \eta)}{\gamma(\ve)},\; x(\ve, \eta) + C \frac{\Gamma (\ve, \eta)}{\gamma(\ve)}\bigr]
\subset (-b/2,b/2)
\end{equation}
where
\begin{equation}\label{Gammaveeta}
 \Gamma (\ve, \eta):=-I(x(\ve, \eta), \ve, \eta).
\end{equation}
Define
\begin{equation}\label{Lveeta}
L(x, \ve, \eta):=-(x- x(\ve, \eta))-i\Gamma (\ve, \eta).
\end{equation}
The technical core of the proof is to show that uniformly in $t\in \bR$ and $\eta \searrow 0$ we have
\begin{equation}\label{mainest}
 \Bigl|\int_{J_{\ve,\eta}}e^{-ixt}\Bigl(\frac{1}{ F(x+i\eta, \ve)}-\frac{1}{L(x, \ve, \eta )}\Bigr)dx \Bigr|\lesssim \gamma(\ve). 
\end{equation}
Let us prove this. By construction we have $L(x(\ve, \eta), \ve, \eta)= F(x(\ve, \eta)+i\eta, \ve)$, thus:
\begin{align}\label{L-F}
L(x, \ve, \eta)-F(x+i\eta, \ve)&=\int_{x(\ve, \eta)}^x\frac{d}{du}( L(u, \ve, \eta)-F(u+i\eta, \ve))du \notag\\
&=(x-x(\ve, \eta))\frac{d}{du}( L(u, \ve, \eta)-F(u+i\eta, \ve))|_{u=x(\ve, \eta)}\notag\\
&\quad+\int_{x(\ve, \eta)}^x\Bigl\{\frac{d}{du}( L(u, \ve, \eta)-F(u+i\eta, \ve))\notag\\
&\qquad-\frac{d}{du}( L(u, \ve, \eta)-F(u+i\eta, \ve))|_{u=x(\ve, \eta)}\Bigr \}du.
\end{align}
From \eqref{Fz} and \eqref{Lveeta} we get
\begin{equation}\label{L-F2}
 \Bigl|\frac{d}{du}( L(u, \ve, \eta)-F(u+i\eta, \ve))|_{u=x(\ve, \eta)}\Bigr| \lesssim \gamma(\ve)
\end{equation}
and  from\eqref{Fz}, \eqref{Lveeta} and (v) we then get
\begin{align}\label{L-F3}
\Bigl|\int_{x(\ve, \eta)}^x\Bigl\{\frac{d}{du}( L(u, \ve, \eta)&-F(u+i\eta, \ve))\notag\\
&-\frac{d}{du}( L(u, \ve, \eta)-F(u+i\eta, \ve))|_{u=x(\ve, \eta)}\Bigr\}du \Bigr|\notag\\ 
 &\qquad
 \lesssim \gamma(\ve)\int_{x(\ve, \eta)}^x\omega(|u-x(\ve, \eta)|)du
 \notag\\
 &\qquad \lesssim \gamma(\ve) |x-x(\ve, \eta)|\omega(|x-x(\ve, \eta)|).
\end{align}
Now we write
\begin{align}\label{L/F}
 \frac{1}{F(x+i\eta, \ve)}-\frac{1}{ L(x, \ve, \eta)}&=\frac{ L(x, \ve, \eta)-F(x+i\eta, \ve)}{ L(x, \ve, \eta)^2}\notag\\
 &\quad+
\frac{\bigl(L(x, \ve, \eta)-F(x+i\eta, \ve)\bigr)^2}{ L(x, \ve, \eta)^2F(x+i\eta, \ve)},
\end{align}
and estimate the terms in the r.h.s. of \eqref{L/F}. From the first equality in \eqref{L-F}, \eqref{Lveeta}, \eqref{L-F2}
and the fact that
\begin{equation}\label{F>}
 |x-x(\ve, \eta)| \lesssim |R(x, \ve, \eta)| \leq |F(x+i\eta, \ve)|,
\end{equation}
we have
\begin{equation*}
 \Bigl|\frac{( L(x, \ve, \eta)-F(x+i\eta, \ve))^2}{ L(x, \ve, \eta)^2F(x+i\eta, \ve)}\Bigr| \lesssim 
\gamma(\ve)^2\frac{|x-x(\ve, \eta)|}{|x-x(\ve, \eta)|^2+\Gamma (\ve, \eta)^2}
\end{equation*}
and then from \eqref{interval} (remember that $\lim_{\ve \searrow 0}\gamma(\ve)=0$):
\begin{align}\label{Con1}
 \Bigl|\int_{J_{\ve,\eta}}&\frac{( L(x, \ve, \eta)-F(x+i\eta, \ve))^2}{ L(x, \ve, \eta)^2F(x+i\eta, \ve)}dx\Bigr| \notag\\
 &\lesssim
\gamma(\ve)^2\int_0^{C\frac{\Gamma (\ve, \eta)}{\gamma(\ve)}}\frac{y}{y^2+\Gamma (\ve, \eta)^2}dy 
\lesssim \gamma(\ve)^2 \ln(1+\gamma(\ve)^{-2}).
\end{align}
We now come to the first term in the r.h.s. of \eqref{L/F}. We claim that
\begin{equation}\label{est1}
\Bigl| \int_{J_{\ve,\eta}}e^{-ixt}\frac{x(\ve, \eta)-x}{ L(x, \ve, \eta)^2}dx \Bigr| \lesssim 1.
\end{equation}
Indeed, write
\begin{equation}\label{x/L}
 \frac{x(\ve, \eta)-x}{ L(x, \ve, \eta)^2}=\frac{1}{ L(x, \ve, \eta)} +\frac{i\Gamma (\ve, \eta)}{ L(x, \ve, \eta)^2}.
\end{equation}
Using
$ \frac{1}{\pi}\int_{\bR}\frac{\Gamma (\ve, \eta)}{x^2 +\Gamma (\ve, \eta)^2}dx =1$ one has
\begin{equation}\label{est2}
 \Bigl| \int_{J_{\ve,\eta}}e^{-ixt}\frac{x(\ve, \eta)-x}{ L(x, \ve, \eta)^2}dx \Bigr| \leq 
\Bigl| \int_{J_{\ve,\eta}}e^{-ixt}\frac{1}{L(x, \ve, \eta)}dx \Bigr|
+\pi.
\end{equation}
Further (see \eqref{interval})
\begin{equation}\label{est3}
\Bigl| \int_{J_{\ve,\eta}}e^{-ixt}\frac{1}{L(x, \ve, \eta)}dx\Bigr|=\Bigl |\int_{-C/\gamma(\ve)}^{ C/\gamma(\ve)}\frac{e^{-it\Gamma (\ve, \eta)u}}{u+i}
du \Bigr| \lesssim 1,
\end{equation}
where the last inequality is obtained by closing the contour in the upper complex plane with a semicircle. Putting together \eqref{x/L}, \eqref{est2},
\eqref{est3} one obtains \eqref{est1}.

From \eqref{est1} and \eqref{L-F2}:
\begin{align}\label{Con2}
\Bigl| \int_{J_{\ve,\eta}}e^{-ixt} \frac{(x-x(\ve, \eta))[\frac{d}{dx}( L(x, \ve, \eta)-F(x+i\eta, \ve))]_{x=x(\ve, \eta)}}
{ L(x, \ve, \eta)^2}&dx \Bigr|\notag\\
&\lesssim \gamma(\ve)
\end{align}
while from \eqref{L-F3}:
\begin{align}\label{Con3}
\Bigl| \int_{J_{\ve,\eta}}e^{-ixt}\frac{1}{ L(x, \ve, \eta)^2} \int_{x(\ve, \eta)}^x\Bigl\{&\frac{d}{du}\bigl( L(u, \ve, \eta)-F(u+i\eta, \ve)\bigr)\notag\\
&-\frac{d}{du}
\bigl( L(u, \ve, \eta)-F(u+i\eta, \ve)\bigr)\bigr|_{u=x(\ve, \eta)}\Bigr\}du dx \Bigr|\notag\\ 
&\leq \gamma(\ve)\int_{J_{\ve,\eta}}\frac{|x-x(\ve, \eta)|\omega(|x-x(\ve, \eta)|)}{|L(x, \ve, \eta)^2|}dx\notag\\
&\lesssim \gamma(\ve).
\end{align} 
Finally,  gathering together \eqref{L/F},\eqref{Con1}, \eqref{L-F}, \eqref{Con2} and \eqref{Con3} one obtains \eqref{mainest}.

By taking the complex conjugate in \eqref{mainest} and replacing $t$ with $-t$ one also obtains
\begin{equation}\label{mainestcc}
 \Bigl|\int_{J_{\ve,\eta}}e^{-ixt}(\frac{1}{ \overline{F(x+i\eta, \ve)}}-\frac{1}{\overline{L(x, \ve, \eta)}})dx \Bigr|\lesssim \gamma(\ve). 
 \end{equation} 
We now claim that if $t\geq 0$ we have
\begin{equation}\label{Lint}
 \Bigl|\frac{1}{\pi}\int_{J_{\ve,\eta}}e^{-ixt}\frac{\Gamma (\ve, \eta)}{(x-x(\ve, \eta))^2 +\Gamma (\ve, \eta)^2}dx-e^{-it(x(\ve, \eta)-i\Gamma (\ve, \eta))}\Bigr|
 \lesssim \gamma(\ve).
 \end{equation}
 Indeed, by direct computation we get
\begin{align}\label{CLint}
 \Bigl|\Bigl(\int_{\bR}-\int_{J_{\ve,\eta}}\Bigr)e^{-ixt}&\frac{\Gamma (\ve, \eta)}{(x-x(\ve, \eta))^2 +\Gamma (\ve, \eta)^2}dx\Bigr|\notag\\
 &\leq 2\int_{C\frac{\Gamma (\ve, \eta)}{\gamma(\ve)}}^\infty 
 \frac{\Gamma (\ve, \eta)}{x^2 +\Gamma (\ve, \eta)^2}dx \lesssim \gamma(\ve).
 \end{align}
 On the other hand, by exact integration, we get
\begin{equation*}
 \frac{1}{\pi}\int_{\bR}e^{-ixt}\frac{\Gamma (\ve, \eta)}{(x-x(\ve, \eta))^2 +\Gamma (\ve, \eta)^2}dx=e^{-it(x(\ve, \eta)-i\Gamma (\ve, \eta))},
 \end{equation*}
which together with \eqref{CLint} gives \eqref{Lint}.

From \eqref{mainest}, \eqref{mainestcc}, and \eqref{CLint} one obtains for all $t\geq0$
\begin{multline}\label{mainest1}
\Bigl|\frac{1}{2\pi i}\int_{J_{\ve \eta}}e^{-ixt}\Bigl(\frac{1}{F(x+i\eta, \ve)} \\-\frac{1}{ \overline{F(x+i\eta, \ve)}}\Bigr)dx -
e^{-it(x(\ve, \eta)-i\Gamma (\ve, \eta))}
\Bigr|\lesssim \gamma(\ve). 
 \end{multline}

 We now finish the proof of Lemma~\ref{Techest} by using a trick going back to Hunziker \cite{Hu}.  Write
 \begin{align*}
I(\ve, \eta, t)&=\frac{1}{\pi} \int_{\bR}e^{-itx}\frac{\im F(x +i\eta, \ve)}{|F(x +i\eta, \ve)|^2}dx\\ 
&=\frac{1}{\pi}\Bigl(\int_{J_{\ve,\eta}}+\int_{\bR \setminus J_{\ve,\eta}}\Bigr )e^{-itx}\frac{\im F(x +i\eta, \ve)}{|F(x +i\eta, \ve)|^2}dx\\
&=:I_1(\ve, \eta, t)+I_2(\ve, \eta, t).
\end{align*}
By assumption (iv) we have that
\begin{equation}\label{ass}
\lim_{\eta \searrow 0} I(\ve, \eta, 0)=1,
\end{equation}
while from \eqref{mainest1}, uniformly in $\eta \searrow 0$,
\begin{equation}
 |I_1(\ve, \eta, 0)-1| \lesssim \gamma(\ve).
\end{equation}
It follows that, uniformly in $\eta \searrow 0$
\begin{align}
|I_2(\ve, \eta, 0)|&\leq |I(\ve, \eta, 0)-1|+ |I_1(\ve, \eta, 0)-1| \notag\\
 &\lesssim
 |I(\ve, \eta, 0)-1|+\gamma(\ve),
\end{align}
and then from \eqref{ass} and the fact that $|I_2(\ve, \eta, t)|\leq |I_2(\ve, \eta, 0)|$:
\begin{equation}\label{I2t}
 \limsup_{\eta \searrow 0}|I_2(\ve, \eta, t)| \lesssim \gamma(\ve).
\end{equation}
Finally, using \eqref{I2t} and the fact that 
$$\lim_{\eta \searrow 0}\bigl(x(\ve, \eta)-i\Gamma (\ve, \eta)\bigr)=E_{\ve},$$
we obtain for an arbitrary $t\geq 0$
\begin{align*}
\limsup_{\eta \searrow 0}|I(\ve, \eta, t)-e^{-itE_{\ve}}| &\leq \limsup_{\eta \searrow 0}|I_1(\ve, \eta, t)-e^{-it
(x(\ve, \eta)-i\Gamma (\ve, \eta)) }|\\
&\quad+\limsup_{\eta \searrow 0}|I_2(\ve, \eta, t)| \\
&\lesssim \gamma(\ve),
\end{align*}
and the proof is finished.

\section{Application to two-channel Schr\"{o}dinger operators}\label{section4}
We apply the abstract theory developed in the previous sections to a certain class of   two-channel Schr\"{o}dinger operators in arbitrary dimensions, 
as they appear for example in the theory of Feshbach resonances in atomic physics; see e.g. \cite{NSJ,KGJ} and references given there. The model
has been considered also in \cite{JN2} and \cite{DJN1} in connection with FGR at thresholds. We follow the setting and notation in Section 5 of 
\cite{JN2}.
 
Our two-channel Schr\"{o}dinger operator has a non-degenerate bound state in
the `closed' channel, whose Hilbert space is modelled with $\mathbf{C}$. The total Hilbert space is $\mathcal{H}=L^{2}(\mathbf{R}^{d}) \oplus \mathbf{C}$. As the unperturbed Hamiltonian we take
\begin{equation}\label{H2}
    H=\begin{bmatrix}
  -\Delta +V -E_0& 0 \\
  0 & 0
\end{bmatrix},\quad  E_0 >0,
\end{equation}
where $V$ satisfies for some $\gamma >0 $ to be specified later
\begin{equation}\label{V}
    \jap{\cdot}^{\gamma}V \in L^{\infty}(\bR^{d}).
\end{equation}
Here $\jap{x}=(1+x^2)^{1/2}$ as usual.
The perturbation is
\begin{equation}\label{Wtc}
W=\begin{bmatrix}
  W_{11} & \ket{W_{12}}\bra{1} \\
  \ket{1}\bra{W_{12}} & b
\end{bmatrix},
\end{equation}
which is a shorthand for
\begin{equation}\label{W1}
W\begin{bmatrix}
  f(x) \\
  \xi 
\end{bmatrix}
=
\begin{bmatrix}
  W_{11}(x)f(x)+W_{12}(x)\xi \\[5pt]
  \int_{\bR^d}\overline{W_{12}(x)}f(x)dx+b\xi 
\end{bmatrix}.
\end{equation}
Here we assume
\begin{equation}\label{2W}
\jap{\cdot}^{\gamma}W_{11}
\in L^{\infty}(\mathbf{R}^{d}), 
\quad 
\jap{\cdot}^{\gamma/2}W_{12}
\in L^{\infty}(\mathbf{R}^{d}), 
\end{equation}
and furthermore that $W_{11}$ is  real-valued and $b \in \bR$ .
 We introduce the weight function
\begin{equation}\label{weight}
\rg=\jap{\cdot}^{-\gamma/2}
\end{equation}
the weight operator
\begin{equation}\label{B}
    B=\begin{bmatrix}
  \rmg & 0 \\
  0 & 1
\end{bmatrix}
,
\end{equation}
and define the bounded self-adjoint operator, $C$,  and its polar decomposition with $D=D^* =D^{-1}$
\begin{equation}\label{C}
C=BWB=|C|^{1/2}D|C|^{1/2}.
\end{equation}
In the factorization (see \eqref{fact}) of $W$ we take $\cK=\cH$ and
\begin{equation}\label{A2}
A=|C|^{1/2}B^{-1},
\end{equation}
 so that
\begin{equation}\label{fW}
    W=B^{-1}|C|^{1/2}D|C|^{1/2}B^{-1}.
\end{equation}
Notice that in our case
 \begin{equation}\label{PQ00}
    P_{0}=\begin{bmatrix}
  0 & 0 \\
  0 & 1 
\end{bmatrix},\quad 
 Q_{0}= 1-P_0=\begin{bmatrix}
  1 & 0 \\
  0 & 0 
\end{bmatrix}
,
\end{equation}
i.e. $Q_0$ is the orthogonal projection onto $L^{2}(\mathbf{R}^{d})$. The key point of the above factorization is that
\begin{equation}\label{com}
 [B,Q_0]=0,
\end{equation}
so that  (see \eqref{w}):
\begin{align}\label{wtc}
G(z)&=AQ_{0}(H-z)^{-1}Q_{0}A^{*}\notag\\
&=|C|^{1/2}\begin{bmatrix}
\rg(-\Delta +V-E_0-z)^{-1}\rg & 0\\
0&0 
\end{bmatrix}
 |C|^{1/2}.                                  
\end{align}

Concerning Assumption \ref{SM}, notice that in our case
\begin{equation}
 Q(J_a)=
 \begin{bmatrix}
  E(J_a) & 0 \\
  0 & 0 
\end{bmatrix}
,
\end{equation}
where $E(\Delta)$ is the spectral measure of $-\Delta +V$ in $L^{2}(\mathbf{R}^{d})$. Thus the verification of Assumption \ref{SM} boils down to the verification of the corresponding assumption for
$-\Delta +V-E_0$ in $L^{2}(\mathbf{R}^{d})$ with $A^*$ replaced by $\jap{\cdot}^{-\gamma/2}$ and $\Gamma(\lambda)$ replaced by 
the corresponding operator for $-\Delta +V$ in $L^{2}(\mathbf{R}^{d})$.

We summarize the above discussion in the following lemma.
\begin{lemma}\label{Asstc}
\begin{itemize}
\item[\rm(i)] Assumption \ref{smG} is implied by
 \begin{equation}
 \rg(-\Delta +V-E_0-z)^{-1}\rg \in C^{1,\omega}(D_a;\cB(L^2(\bR^d))) 
 \end{equation}
and $\int_0\frac{\omega(x)}{x}dx <\infty$.

\item[\rm(ii)] Assumption  \ref{SM} is implied by
$$ \Gamma(\cdot) \jap{\cdot}^{-\gamma/2}\in 
C^{2,\omega_{\theta}}(J_a;\cB(L^2(\bS^{d-1})))$$
and 
 for some $\theta \in (0,1)$. Here $\bS^{d-1}$ is the unit sphere in $\bR^d$ with the induced measure 
 and $\Gamma(\cdot)$ is the trace operator
of $-\Delta +V$ corresponding to $J_a$.
\end{itemize}
\end{lemma}
Both these facts are well known in spectral theory of $d$-dimensional Schr\"{o}d\-inger operators with rapidly decaying potential  and for 
the convenience of the reader we shall recall some of these results in Appendix~\ref{appendix2}. In particular, from these results 
one has that  the conclusion of Theorem~\ref{JN06} holds true for $\gamma >3$ and if the FGR constant vanishes  the conclusion of Theorem~\ref{ve4} 
holds true for $\gamma >5$.

\appendix
\section{H\"older continuity for Cauchy integral transform}\label{appendix1}

The Cauchy integral transform preserves $\theta$-H\"older continuity for $\theta \in (0,1)$. Even though the result is known, we prove it here in a form which is appropriate for our needs. The argument is a generalization of the one in \cite[Ch. I, \S 5]{Gak}.

Let $\Phi(\cdot)$ be a complex valued function satisfying
\begin{equation}\label{Phi}
 {\rm supp} \;\Phi \subset (-1,1), \quad\Phi \in C^{n, \omega_{\theta}}(\bR;\bC),\; \theta \in  (0,1).
\end{equation}
There exists a constant $C_{n,\theta}$ such that:
$$|\Phi^{(n)}(x)-\Phi^{(n)}(y)|\leq C_{n,\theta}\; |x-y|^\theta,\quad \forall x,y\in(-1,1). $$
Define 
$$\||\Phi|\|_{n,\theta}:=\max\{\|\Phi\|_\infty, \|\Phi^{(1)}\|_\infty,\dots, \|\Phi^{(n)}\|_\infty, C_{n,\theta}\}.$$
Define for $z=x+i\eta \in \bC \setminus [-1,1]$ the Cauchy transform 
\begin{equation}\label{Psi}
 \Psi (z)= \int_{\bR}\frac{\Phi(x)}{x-z}dx.
\end{equation}
\begin{proposition}\label{G}
The map $\Psi$ is holomorphic on $\bC \setminus [-1,1]$. For every $k=0,1,...,n$  the limits  
$\Psi^{(k)} (x)= \lim_{\eta \searrow 0}\Psi^{(k)} (x+i\eta)$ exist. Moreover, there exists a constant $C$ such that uniformly in $x,y \in (-2,2)$, $\eta>0$ and $0\leq k\leq n$ we have:
\begin{equation}\label{H}
|\Psi^{(k)} (x+i\eta)-\Psi^{(k)} (y+i\eta)| \leq C \;\||\Phi|\|_{n,\theta}\; |x-y|^{\theta}.
\end{equation}
\end{proposition}

\begin{proof}
A finite number of constants appearing during the proof will be denoted by $C$.
If $z\notin [-1,1]$ we have $\Psi^{(k)}(z)=\int_{\bR}\frac{\Phi(x)}{(x-z)^{k+1}}dx$; integrating by parts, we can write 
$\Psi^{(k)}(z)=\int_{\bR}\frac{\Phi^{k}(x)}{x-z}dx$, thus it is sufficient to prove the proposition for $n=0$.
The argument  for the existence of limit values is the standard principal value argument and it will not be repeated here.
The argument for H\"older continuity is more elaborated. Consider
\begin{align*}
 \widetilde{\Psi}(x+i\eta)&=\int_{-11}^{11}\frac {\Phi(\tau)-\Phi(x)}{\tau-x-i\eta}d\tau\\
 &=
 \Psi(x+i\eta)-\Phi(x)\ln \frac{11-x-i\eta}{-11-x-i\eta}.
\end{align*}
Since the second term in the r.h.s. satisfies \eqref{H}, it is sufficient to consider $\widetilde{\Psi}(x+i\eta)$. In what follows, 
$x,y \in (-2,2)$ and $ \eta \in (0,1)$. 

Denote by $L:=[-11,11]$. For a given pair $x_1<x_2$  in $(-2,2)$,  we define 
$$l := (x_1-2|x_1-x_2 |,x_2+2|x_1-x_2 |)
=:(a,b)\subset [-10,10]\subset L.$$

For  $z_j =x_j+i\eta$, $j=1,2$ we have to estimate $\widetilde{\Psi}(z_2)-\widetilde{\Psi}(z_1)$. We write
\begin{align}\label{Psi12}
\widetilde{\Psi}(z_2)-\widetilde{\Psi}(z_1)&= \int_l \frac {\Phi(\tau)-\Phi(x_2)}{\tau-z_2}d\tau -
\int_l \frac {\Phi(\tau)-\Phi(x_1)}{\tau-z_1}d\tau\notag \\ 
&\quad+\int_{L\setminus l} \Bigl(\frac {\Phi(\tau)-\Phi(x_2)}{\tau-z_2}-\frac {\Phi(\tau)-\Phi(x_1)}{\tau-z_1}\Bigr)d\tau.
\end{align}
The first two integrals are easily estimated
\begin{align}\label{int1}
 \Bigl| \int_l \frac {\Phi(\tau)-\Phi(x_2)}{\tau-z_2}d\tau \Bigr| & \leq C\int_l \frac{|\tau-x_2|^{\theta}}{|\tau-x_2|}d\tau\notag\\
 & \leq C\int_0^{3|x_1-x_2|}u^{\theta -1}d\tau = C\frac{|x_1-x_2|^{\theta}}{\theta},
\end{align}
and similarly for the second one. In the third integral of \eqref{Psi12} one uses the following identity
\begin{align}\label{idphi}
\frac{\Phi(\tau)-\Phi(x_2)}{\tau-z_2}&-\frac {\Phi(\tau)-\Phi(x_1)}{\tau-z_1}\notag\\
&=\frac {\Phi(x_1)-\Phi(x_2)}{\tau-z_1}+ \frac {(\Phi(\tau)-\Phi(x_2))(z_2-z_1)}{(\tau-z_1)(\tau-z_2)}.
\end{align}
For the integral involving the first term in the r.h.s. of \eqref{idphi} we observe that $|a-z_1|=|b-z_1|$  and
\begin{align*}
 \int_{L\setminus l}\frac{1}{\tau-z_1}d\tau= \ln\Bigl |\frac{11-z_1}{11+z_1}\Bigr|+ i\arg\Bigl (\frac{11-z_1}{z_1-b}\frac{a-z_1}{11+z_1}\Bigr),
\end{align*}
which is uniformly bounded for $x_1 \in (-2,2)$ and $ \eta \in (0,1)$, hence
\begin{equation}\label{int2}
  \Bigl|\int_{L\setminus l}\frac{\Phi(x_1)-\Phi(x_2)}{\tau-z_1}d\tau \Bigr| \leq C\; |x_1-x_2|^{\theta}.
\end{equation}
We are left  with estimating the integral involving the second term in the r.h.s. of \eqref{idphi}:
\begin{align}\label{int3}
 \Bigl|\int_{L\setminus l}&\frac{(\Phi(\tau)-\Phi(x_2))(z_2-z_1)}%
 {(\tau-z_1)(\tau-z_2)}d\tau \Bigr|\notag\\
 &\leq
 C|x_1-x_2|\int_{L\setminus l}\frac{1}{|\tau-x_1\|\tau-x_2|^{1-\theta}}
 d\tau \notag\\ 
 &=C |x_1-x_2|\int_{L\setminus l}\Bigl|\frac{\tau-x_1}{\tau-x_2}\Bigr |^{1-\theta}\frac{1}{|\tau-x_1|^{2-\theta}}d\tau.
\end{align}

Because $\frac{\tau-x_1}{\tau-x_2}$ is piecewise monotone as a function of $\tau$, we have
$$
\sup_{x_1 \in (-2,2), \tau \in L\setminus l}\Bigl |\frac{\tau-x_1}{\tau-x_2}\Bigr | < \infty,
$$
where the maximum is attained in the set $ \{-11, 11, a, b\}$.
Using this in \eqref{int3} one obtains
\begin{align}\label{int4}
|x_1-x_2|\int_{L\setminus l}
&\frac{|\tau-x_1|}{|\tau-x_2|}^{1-\theta}\frac{1}
{|\tau-x_1|^{2-\theta}}d\tau \notag\\
 &\leq C |x_1-x_2|
 \int_{2|x_1-x_2|}^{11}u^{\theta -2}d\tau \leq \frac{C}{1- \theta}|x_1-x_2|^{\theta}.
\end{align}
Putting together \eqref{Psi12}, \eqref{int1} - \eqref{int4}, the proof is finished.
\end{proof}

\section{Resolvent smoothness and $\Gamma$ operator for one body $d$-dimensional Schr\"odinger operators}\label{appendix2}

A convenient formalism for stationary scattering theory is given by the 
$\Gamma$ operators or trace operators. These were introduced in \cite{kuroda} in an abstract setting. A presentation of applications to Schr\"{o}dinger operators can be found in \cite{kuroda-aarhus}. Extensive results on stationary scattering theory can be found in \cite{yafaevI,yafaevII}, both in an abstract framework, and applied to a number of differential operators. Trace operators are used in many cases in these monographs.

We now describe the trace operators for Schr\"{o}dinger Hamiltonians $-\Delta+V$ in $L^2(\bR^d)$ for the convenience of the reader.

Let $H_0=-\Delta$ on $\cH=L^2(\bR^d)$, with the domain $\cD(H_0)=H^2(\bR^d)$, the usual Sobolev space. 

We  need the weighted $L^2$-spaces. We have
\begin{equation}
L^{2,s}(\bR^d)=\{f\,|\, \jap{\cdot}^sf\in L^2(\bR^d)\},\quad s\in\bR.
\end{equation}
Here $\jap{\cdot}$ denotes multiplication by $\jap{x}=(1+x^2)^{1/2}$, $x\in\bR^d$. We use the following convention for the Fourier transform.
\begin{equation}
\cF\colon L^2(\bR^d)\to L^2(\bR^d),\quad 
(\cF f)(\xi)=\hat{f}(\xi)=\frac{1}{(2\pi)^{d/2}}\int e^{-ix\xi}f(x)dx.
\end{equation}
We also use the Fourier transform between other spaces. For example we have 
$\cF(L^{2,s}(\bR^d))=H^s(\bR^d)$, $s\in\bR$, the Sobolev spaces.

We let $J=(0,\infty)$.

\begin{definition}
The free $\Gamma$ operator is defined for $f\in L^{2,s}(\bR^d)$, $s>\frac12$, as
\begin{equation}\label{B.3}
(\Gamma_0(\lambda)f)(\omega)=2^{-1/2}\lambda^{(d-2)/4}(\cF f)(\lambda^{1/2}\omega),\quad \lambda\in J, \; \omega\in\bS^{d-1}.
\end{equation}
\end{definition}
We record some of the properties of $\Gamma_0$. We use the notation $\sh=L^2(\bS^{d-1})$. We also use the notation $\bsh$ for the bounded operators. Furthermore, we use the H\"older space 
\begin{equation}
C^{n,\theta}(J,\bsh), \quad n\in\bN,\; 0<\theta<1,
\end{equation}
and also the local H\"older space $C^{n,\theta}_{\rm loc}(J,\bsh)$, which means that the functions are H\"older continuous on any relatively compact open subinterval of $J$.
\begin{proposition}
For $\lambda\in J$ and $s>\frac12$
we have $\Gamma_0(\lambda)\in\bsh$. If $s=\frac12+n+\theta$, $n\in\bN$, $0<\theta<1$, then $\Gamma_0\in C^{n,\theta}_{\rm loc}(J,\bsh)$. If $s=\frac12+n+1$, then
$\Gamma_0\in C^{n,\theta}_{\rm loc}(J,\bsh)$ for all $0<\theta<1$.
\end{proposition}
The result follows from the fact that $\cF(L^{2,s}(\bR^d))=H^s(\bR^d)$ and the trace theorem in Sobolev spaces, see for example \cite{adams,yafaevII}.

Using $\Gamma_0$ one then defines the spectral representation of $H_0$ as follows.
\begin{definition}
The spectral representation of $H_0$ is defined for $f\in L^{2,s}(\bR^d)$, $s>\frac12$, by 
\begin{equation}\label{B.5}
(\cF_0f)(\lambda)(\omega)=(\Gamma_0(\lambda)f)(\omega), \quad \lambda\in J,\quad \omega\in\bS^{d-1}.
\end{equation}
\end{definition}
\begin{proposition}
$\cF_0$ extends to a unitary map from $\cH$ to $L^2(J,\sh)$. Furthermore, we have $\cF_0 H_0 =M_{\lambda}\cF_0$, where $M_{\lambda}$ is the operator of multiplication by $\lambda$ in $L^2(J,\sh)$.
\end{proposition}

This result is then the starting point for obtaining $\Gamma$ operators for $H=H_0+V$. We only outline some of the basic results. The definition relies on the boundary values of the resolvent $R(z)=(H-z)^{-1}$.

We will not deal with local singularities of the perturbation $V$, so we use the following assumption. The compact operators between two Hilbert spaces are denoted by $C(\cH,\cK)$.
\begin{assumption}\label{SR}
Assume that $V$ is a bounded self-adjoint operator on $\cH$,
such that for some $\beta>\tfrac12$ it satisfies $V\in C(L^{2,-\beta}(\bR^d),L^{2,\beta}(\bR^d))$.
\end{assumption}

\begin{assumption}\label{LAP}
Assume that the limiting absorption principle holds for $R(z)$ on $J$ with boundary values 
\begin{equation*}
R(\lambda\pm i0)\in\cB(L^{2,s}(\bR^d),L^{2,-s}(\bR^d)), \quad\text{for some $s>\tfrac12$}.
\end{equation*}
Assume that the boundary values are locally H\"older continuous with exponent $\theta$, $0<\theta<s-\tfrac12$ and  $\theta<1$.
\end{assumption}

Note that this formulation excludes positive eigenvalues for $H$.
This assumption can be verified in different manners, for example by using Mourre theory, see \cite{GGM,ABG}. Differentiability of the boundary values and H\"older continuity of the highest derivative of the boundary values can also be verified, provided sufficiently strong assumptions are imposed on $V$.

Let us state a special case of the results in  \cite{GGM,ABG}.
\begin{proposition}
Let $V$ be a real-valued function. Assume that there exist $\beta>\frac12$ and $C>0$ such that
\begin{equation}
\abs{V(x)}\leq C\jap{x}^{-2\beta},\quad x\in\bR^d.
\end{equation}
Let $H=H_0+V$. 
If $\beta>\frac12+n+\theta$, $n\in\bN$, $0<\theta<1$, and $s>\beta$, then the boundary values exist and satisfy 
\begin{equation}
R(\cdot\pm i0)\in
C^{n,\theta}_{\rm loc}(J,\cB(L^{2,s}(\bR^d),L^{2,-s}(\bR^d))).
\end{equation}
\end{proposition}

With these preparations we state the following definition.
\begin{definition}
Assume that Assumption~\ref{SR} is verified for some $\beta>\tfrac12$ and that Assumption~\ref{LAP} is satisfied for $s=\beta$. For $f\in L^{2,\beta}(\bR^d)$ we define
\begin{equation}\label{B8}
\Gamma_{\pm}(\lambda)f=\Gamma_0(\lambda)(I-VR(\lambda\pm i0))f.
\end{equation}
\end{definition}
These operators then have the same H\"older continuity properties as $\Gamma_0$, with $s=\beta$.

\begin{definition}
The spectral representations are defined for $f\in L^{2,\beta}(\bR^d)$ by
\begin{equation}
(\cF_{\pm}f)(\lambda)(\omega)=(\Gamma_{\pm}(\lambda)f)(\omega),\quad \lambda\in J,\; \omega\in\bS^{d-1}.
\end{equation}
\end{definition}

We denote the projection onto the absolutely continuous subspace for $H$ by $\pac$. Then we can state the following result.
\begin{proposition}
The operators $\cF_{\pm}$ extend to unitary operators from $\pac\cH$ to $L^2(J,\sh)$.
Furthermore, we have $\cF_{\pm} H =M_{\lambda}\cF_{\pm}$, where $M_{\lambda}$ is the operator of multiplication by $\lambda$ in $L^2(J,\sh)$.
\end{proposition}

Let $E_0(\lambda)$ and $E(\lambda)$ denote the spectral families of $H_0$ and $H$ respectively. Then for $f,g\in L^{2,\beta}(\bR^d)$ we have for $\lambda\in J$
\begin{align}
\ip{f}{E_0^{\prime}(\lambda)g}&=\ip{\Gamma_0(\lambda)f}{\Gamma_0(\lambda)g},\\
\ip{f}{E^{\prime}(\lambda)g}&=\ip{\Gamma_{\pm}(\lambda)f}{\Gamma_{\pm}(\lambda)g}.
\end{align}

This framework is used for the derivation of the stationary scattering for the pair of operators $H_0$ and $H$.

\begin{remark}
The constructions above based on the Fourier transform can be applied to a number of constant coefficient pesudodifferential operators $H_0=f(-i\nabla)$, provided sufficient information is available on the energy surfaces $\{\xi\in\bR^d\,|\,f(\xi)=\lambda\}$. For example, one can construct free $\Gamma$ operators and spectral representations for the free Dirac operator and for the relativistic Schr\"{o}dinger operators $\sqrt{-\Delta+m^2}$, $m\geq0$.

Using a different transform one can also construct the free $\Gamma$ operator and the spectral representation for the free Stark operator $H_0=-\Delta+x\cdot\cE$, see~\cite{yajima} for the details and~\cite{yajima,jensen} for resonances in the analytic continuation framework based on this spectral representation.
\end{remark}

\section*{Acknowledgements}
H. Cornean and A. Jensen were partially supported by the Danish Council for Independent Research $|$ Natural Sciences, Grants 11-106598. G. Nenciu acknowledges support from a VELUX visiting professorship. G. Nenciu thanks the Department of Mathematical Sciences, Aalborg University, for its hospitality.
We thank the referees for their constructive comments.


\end{document}